\documentclass[10pt,a4paper]{article}
\usepackage[left=3.8cm,right=3.4cm,top=3.4cm,bottom=2.5cm]{geometry}
\usepackage{amsmath}
\usepackage{mathtools}
\usepackage{amsfonts}
\usepackage{amssymb}
\usepackage{physics}
\usepackage{braket}
\usepackage{siunitx}
\usepackage{graphicx}
\usepackage{enumerate}
\usepackage{sidecap}
\usepackage [autostyle, english = american]{csquotes}
\MakeOuterQuote{"}
\usepackage{float}
\usepackage[toc,page,title]{appendix}
\usepackage{multirow}
\usepackage{amsthm}
\usepackage{sidecap}
\usepackage{eucal}
\usepackage{comment}
\usepackage{enumerate}
\usepackage{setspace}
\usepackage[super]{nth}
\usepackage{todonotes}
\usepackage{thmtools}
\usepackage{thm-restate}
\usepackage{tikz}
\usetikzlibrary{backgrounds}
\usepackage{comment}

\usepackage[colorlinks,linktocpage]{hyperref}
\usepackage[capitalise]{cleveref}
\Crefname{lemma}{Lemma}{Lemmas}
\Crefname{proposition}{Proposition}{Propositions}
\Crefname{definition}{Definition}{Definitions}
\Crefname{theorem}{Theorem}{Theorems}
\Crefname{conjecture}{Conjecture}{Conjectures}
\Crefname{corollary}{Corollary}{Corollaries}
\Crefname{example}{Example}{Examples}
\Crefname{section}{Section}{Sections}
\Crefname{appendix}{Appendix}{Appendices}
\Crefname{figure}{Fig.}{Figs.}
\Crefname{equation}{Eq.}{Eqs.}
\Crefname{table}{Table}{Tables}
\Crefname{item}{Property}{Properties}
\Crefname{remark}{Remark}{Remarks}
\Crefname{axioms}{Axioms}{Axioms}

\crefformat{footnote}{#2\footnotemark[#1]#3}

\newcommand\I{\mathrm{1}}

\usepackage[backend=bibtex,style=alphabetic,doi=false,isbn=false,url=false,maxbibnames=20,sorting=none]{biblatex}
\bibliography{DualityFramework}

\title{A mathematical framework for quantum Hamiltonian simulation and duality}
\author{Harriet Apel \footnote{\href{mailto:harriet.apel.19@ucl.ac.uk}{harriet.apel.19@ucl.ac.uk}}
\and Toby Cubitt \footnote{ \href{mailto:t.cubitt@ucl.ac.uk}{t.cubitt@ucl.ac.uk}}}
\date{\textit{Department of Computer Science, University College London, UK}}

\begin{document}


\maketitle

\begin{abstract}
 Analogue Hamiltonian simulation is a promising near-term application of quantum computing and has recently been put on a theoretical footing alongside experiencing wide-ranging experimental success.
  These ideas are closely related to the notion of duality in physics, whereby two superficially different theories are mathematically equivalent in some precise sense.
However, existing characterisations of Hamiltonian simulations are not sufficiently general to extend to all dualities in physics.
We give a generalised duality definition encompassing dualities transforming a strongly interacting system into a weak one and vice versa.
We characterise the dual map on operators and states and prove equivalence of duality formulated in terms of observables, partition functions and entropies.
A building block is a strengthening of earlier results on entropy-preserving maps -- extensions of Wigner's celebrated theorem -- to maps that are entropy-preserving up to an additive constant.
  We show such maps decompose as a direct sum of unitary and anti-unitary components conjugated by a further unitary, a result that may be of independent mathematical interest.
\end{abstract}

\newpage

\tableofcontents

\newpage

\section{Introduction}\label{sect Introduction}

Duality is a deep straining running throughout physics.
Any two systems that are related can be described as being "dual", up to the strictest sense of duality where all information about one system is recoverable in the other.
Calculations or predictions in one theory may be simplified by first mapping to the dual theory, given there is a rigorous relationship between the points of interest.
Strong-weak dualities are a common example of this, allowing well-understood perturbation techniques to be leveraged in high energy regimes by considering the dual weak theory \cite{Baxter1989,tHooft:93,Susskind_1995,Montonen1977}.

In the near-term, there is hope of using quantum computers as analogue simulators to study certain physical properties of quantum many-body systems.
In analogue simulation the Hamiltonian of interest, $H(t)$, is engineered with a physical system that is then allowed to time evolve continuously.
This is in contrast to digital simulation where the time evolution is mapped to quantum circuits -- for example via Trotterisation -- which likely requires a scalable, fault tolerant quantum computer \cite{Lloyd1957}.
It is believed that analogue simulators without error correction could be sufficient to study interesting physics and this has seen varying experimental success with trapped ions \cite{Porras}, cold atoms in optical lattices \cite{Jaksch2004}, liquid and solid state NMR \cite{Peng2010}, superconducting circuits \cite{Houck2012} etc.
These artificial systems allow for improved control and simplified measurements compared to in situ materials, providing a promising use for noisy intermediate scale devices.

What it means for one system to "simulate" or "be dual to" another is an important theoretical question, which has only recently begun to be explored.
\cite{Bravyi2014} and later \cite{Cubitt2019} gave formal definitions of simulation.
Cubitt et al. used this framework to demonstrate certain "universal" spin-lattice models that are able to simulate any quantum many-body system by tuning the interaction parameters.
These works consider the strongest possible definition of a duality: all relevant physics is manifestly preserved in the simulator system including measurement outcomes, the partition function and time evolution.
While this strengthens \cite{Cubitt2019}'s main result it rules out potentially interesting scenarios where the relationship between the systems' properties is more subtle.

Of particular interest to physicists are dualities that relate a strongly interacting theory to a weakly interacting one -- so called strong-weak dualities.
These dualities are of particular interest as strongly interacting theories are beyond the reach of perturbation theory are often challenging to analyse, and strongly interacting phenomena are difficult to elucidate.
Strong-weak dualities serve as a valuable tool for addressing these challenges, transforming a strongly interacting system to a dual, weakly interacting system, which is then amendable to perturbation theory.
One notable example from particle physics is the phenomenon of S-duality, which relates electric and magnetic descriptions within specific gauge theories.
Effectively capturing this and other important classes of dualities in physics, necessitates a more general set of mappings than those previously explored for simulation purposes.

The aim of this work is to explore and extend upon a theoretical framework of duality to better unify operationally how dual systems are related.
We significantly generalise the conditions placed on a duality map between operators to allow for operationally valid transformations which crucially include strong-weak dualities.
This direction of relaxation is inspired by considering examples of duality studied in physics including the Kramer-Wannier duality \cite{Baxter1989} and boson-fermion dualities.
We derive a full characterisation of these maps and additionally characterise the map on states -- which were not shown in previous studies.

Having imposed spectral preserving as a property of dual maps, it follows that other physical properties such as partition functions and entropies are necessarily preserved.
A key outcome of this work is the reverse implication: that demanding partition functions (or entropies) are preserved along with convexity is strong enough to predetermine the spectra of the dual maps.
This leads to three different definitions of duality that, while seemingly distinct with different domains of application, are in fact mathematically equivalent and are therefore characterised by the same mathematical structure, which we demonstrate.
Thus a duality relationship on any one of these physical levels implies a consistent duality on the other physical levels.

The characterisation of entropy preserving maps is a topic of interest independent of simulation with various previous work characterising entropy preserving maps by unitary/antiunitary transformations, \cite{Molnar2008,He2012,He2015}.
Whereas the previous characterisations reduce to Wigner’s theorem, by taking a different route connecting to Jordan and C\* algebra techniques, we show that allowing an entropic additive constant is precisely the additional freedom that allows the maps to admit a direct sum of both unitary and antiunitary parts.
See \cref{fg summary} for a summary of the results and where the formal statements are found in the paper.

\begin{figure}[tbp]
\centering
\begin{tikzpicture}
\begin{scope}[on background layer]
         \node [inner sep=0pt] (1) at (0, 0) {\includegraphics[trim={0cm -1cm 0cm 0cm},clip,scale=0.43]{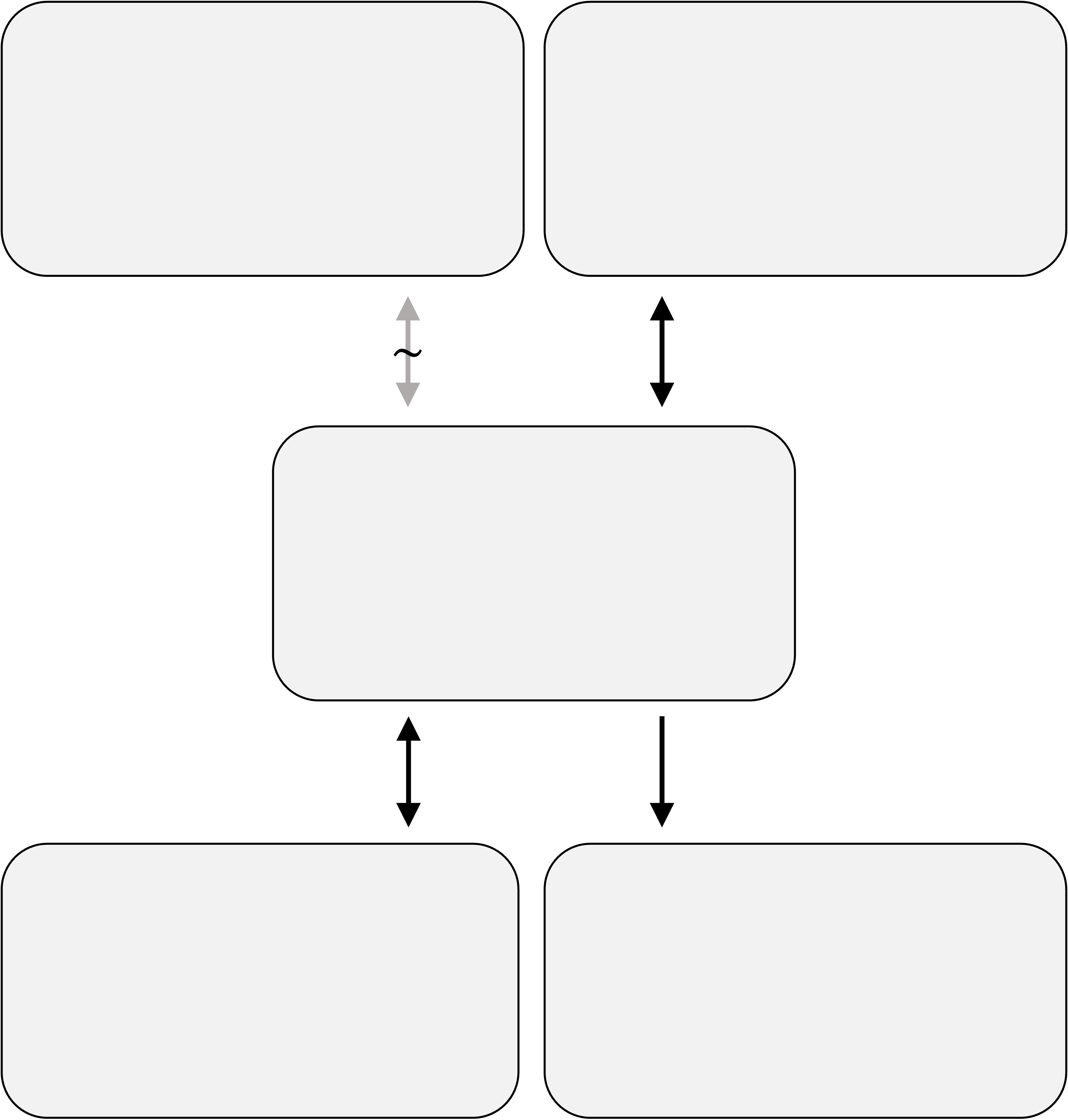}};
\end{scope}
		\node [align=left, anchor=west] (0) at (-2.7, 1.2) {\small Measurement duality map, $\Phi_s$};
		\node [align=left, anchor=west] (0) at (-2.7, 0.8) {\small $\rightarrow$ Maps Hermitians to Hermitians};
		\node [align=left, anchor=west] (0) at (-2.7, 0.4) {\small $\rightarrow$ `Convex' with potential rescalings };
		\node [align=left, anchor=west] (0) at (-2.7, 0) {\small $\rightarrow$ Spectral preserving with potential  };
		\node [align=left, anchor=west] (0) at (-2.2, -0.4) {\small rescaling to include dualities more };
		\node [align=left, anchor=west] (0) at (-2.2, -0.8) {\small broadly};
		\node [align=left, anchor=west] (0) at (0.6, -3.5) {\small A corresponding map on state };
		\node [align=left, anchor=west] (0) at (0.6, -3.9) {\small completes the duality such that};
		\node [align=left, anchor=west] (0) at (0.6, -4.3) {\small measurement outcomes and time };
		\node [align=left, anchor=west] (0) at (0.6, -4.7) {\small dynamics preserved is of the form};
		\node [align=left, anchor=west] (0) at (1, -5.2) {\small $\Phi_\textup{state}(\rho) = U\left(\bigoplus_{i=1}^{p}\alpha_i \rho \right)U^\dagger$ };
		\node [align=left, anchor=west] (0) at (0.6, -5.7) {\small where $\sum_{i=1}^p \alpha_i = 1$.};
		\node [align=left, anchor=west] (0) at (1.6, -2) {\small \cref{prop form of state map}};
		\node [align=left, anchor=west] (0) at (1.6, -2.4) {\small (Map on states)};
		\node [align=left, anchor=west] (0) at (-6, -3.7) {\small The map is of the form };
		\node [align=left, anchor=west] (0) at (-5.8, -4.2) {\small $\Phi(A) = f(A) U \left(A^{\oplus p} \oplus \overline{A}^{\oplus q} \right) U^\dagger$,};
		\node [align=left, anchor=west] (0) at (-6, -4.7) {\small where $p,q$ are non-negative integers, $U$};
		\node [align=left, anchor=west] (0) at (-6, -5.1) {\small is a unitary transformation and $\bar{A}$ };
		\node [align=left, anchor=west] (0) at (-6, -5.5) {\small represents the complex conjugate of $A$.};
		\node [align=left, anchor=west] (0) at (-3.5, -2) {\small \cref{thm Characterisation}};
		\node [align=left, anchor=west] (0) at (-4.4, -2.4) {\small (Characterisation)};
		\node [align=left, anchor=west] (0) at (0.4, 6) {\small Thermal duality map, $\Phi_t$,};
		\node [align=left, anchor=west] (0) at (0.4, 5.6) {\small $\rightarrow$ Maps Hermitians to Hermitians};
		\node [align=left, anchor=west] (0) at (0.4, 5.2) {\small $\rightarrow$ `Convex' with potential rescalings };
		\node [align=left, anchor=west] (0) at (0.4, 4.8) {\small $\rightarrow$ Preserves partition functions with };
		\node [align=left, anchor=west] (0) at (0.9, 4.4) {\small potential rescalings  };
		\node [align=left, anchor=west] (0) at (1.8, 2.8) {\small Equivalent definitions:  };
		\node [align=left, anchor=west] (0) at (1.8, 2.4) {\small \cref{lm partition1}, \cref{col thermal}};
		\node [align=left, anchor=west] (0) at (-6, 6) {\small Entropic duality map, $\Phi_e$,};
		\node [align=left, anchor=west] (0) at (-6, 5.6) {\small $\rightarrow$ Maps Hermitians to Hermitians and };
		\node [align=left, anchor=west] (0) at (-5.5, 5.2) {\small states to states};
		\node [align=left, anchor=west] (0) at (-6, 4.8) {\small $\rightarrow$ Convex };
		\node [align=left, anchor=west] (0) at (-6, 4.4) {\small $\rightarrow$ Preserved entropy of states up to };
		\node [align=left, anchor=west] (0) at (-5.5, 4) {\small state independent additive constant};
		\node [align=left, anchor=west] (0) at (-5.6, 2.8) {\small Closely related definitions:  };
		\node [align=left, anchor=west] (0) at (-5.6, 2.4) {\small \cref{lm entropy spec pres}, \cref{col ax_dual_entr}};
\end{tikzpicture}
\caption{Summary of the main results. We start by considering a duality map $\Phi_s$ that takes in as input observables in one system and outputs the corresponding dual observable in the dual system. There are three constraints that define the map that are be physically motivated, where importantly we allow the map to preserve the eigenspectra (corresponding to measurement outcomes) up to a rescaling -- this allows the map to encompass strong-weak dualities. The main contributions consist of: providing a full mathematical characterisation of these generalised maps (\cref{thm Characterisation}) where $f(A)$ is an operator dependent rescaling function; showing the form of a consistent map on states is implied by the definition of the operator map \cref{prop form of state map}; demonstrating the equivalence of thermal dualities and spectral preserving dualities \cref{col thermal}; relating entropic dualities to spectral preserving dualities \cref{col ax_dual_entr} to give a new characterisation of entropy preserving maps (see \cref{sect Wigners} for discussion).}
 \label{fg summary}
\end{figure}

The following section of this paper gives an overview of key previous works related to the theory of analogue simulation.
Our generalised definition of a duality map is described and characterised in \cref{sect Generalised duality map} with a corresponding map on states.
We then show the equivalence of different duality definitions in \cref{sect equivalent definitions of duality}, highlighting the new characterisation of entropy preserving maps.
Finally we complete the framework by considering errors in the duality map and demonstrating how the framework translates approximate maps to well controlled errors in physical quantities.

%

\subsection{Previous Work}\label{sect Previous Work}
This work uses some results and techniques from \cite{Cubitt2019} in order to build up a more general framework.
This section gives a brief overview of some key results and definitions that are relevant to our investigation, highlighting the constraints that this work will extend.

Encoding maps, denoted $\mathcal{E}$, are at the core of \cite{Cubitt2019}'s simulations.
These maps encode all observables, $A$, on the target Hamiltonian system as observables, $A' = \mathcal{E}(A)$, on the simulator Hamiltonian system and are the most restrictive simulations concerning Hamiltonians in a finite-dimensional Hilbert space.
The authors give a long list of operational requirements that the encoding map should satisfy to exactly reproduce all physical properties of the target system in the absence of errors:
\begin{enumerate}[I]
\item Any observable, $A$ on the target system corresponds to an observable on the simulator system so the map must preserve Hermiticity, $\mathcal{E}(A)=\mathcal{E}(A)^\dagger$;
\item $\mathcal{E}(A)$ preserves the outcomes, and therefore eigenvalues, of any measurement $A$: $\text{spec}[\mathcal{E}(A)]=\text{spec}[A]$;
\item The encoding is real linear, $\mathcal{E}\left(\sum_i \alpha_i h_i \right) = \sum_i \alpha_i \mathcal{E}(h_i)$, for $\alpha_i \in \mathbb{R}$, $h_i \in \text{Herm}$ so that individual Hamiltonian interactions are encoded separately;
\item Measurements are correctly simulated, hence a corresponding map on states, $\mathcal{E}_\text{state}$, should exist such that $\tr\left[\mathcal{E}(A)\mathcal{E}_\text{state}(\rho) \right] = \trace \left[ A \rho \right]$ for all target observables $A$;
\item The encoding preserves the partition function up to a physically unimportant constant rescaling ($c$): $Z_{H'}(\beta) = \trace\left[e^{-\beta \mathcal{E}(H)} \right] = c \trace \left[e^{-\beta H} \right] = c Z_{H}(\beta)$;
\item Time evolution is correctly simulated: $e^{-i\mathcal{E}(H)t}\mathcal{E}_\text{state}(\rho)e^{i\mathcal{E}(H)t} = \mathcal{E}_\text{state}(e^{iHt}\rho e^{iHt})$.
\end{enumerate}

Note the trivial relationships between the physical observables in the simulator and target systems in II-VI, excluding strong-weak dualities.
\cite{Cubitt2019} showed that imposing just three operationally motivated conditions on the encoding will necessarily imply that I-VI hold.
Furthermore, using Jordan and $C^*$~algebra techniques a mathematical characterisation of encodings was given in the following theorem.\footnote{This theorem quoted here is informal and some technicalities have been omitted that can be found in the reference. }

\begin{restatable}[Characterising encodings; see \cite{Cubitt2019} Theorem 4]{thm}{thmEncodings}
\label{thm Encodings}
An \emph{encoding} map $\mathcal{E}$ from Hermitian $(n\times n)$ matrices ($\textup{Herm}_n$) to Hermitian $(m \times m)$ matrices satisfies the following constraints for all $A,B\in \textup{Herm}_n$, and all $p\in [0,1]$ :
\begin{enumerate}[1.]
\item $\mathcal{E}(A)=\mathcal{E}(A)^\dagger$
\item $\textup{spec}[\mathcal{E}(A)]= \textup{spec}[A]$
\item $\mathcal{E}(pA+(1-p)B)= p \mathcal{E}(A) + (1-p)\mathcal{E}(B)$
\end{enumerate}
Encodings are necessarily of the form:
\begin{equation}\label{eqn encoding form}
\mathcal{E}(M)=U(M^{\oplus p}\oplus \overline{M}^{\oplus q})U^\dagger
\end{equation}
for some non-negative integers $p$, $q$ and unitary $U\in \mathcal{M}_m$, where $M^{\oplus p}:=\bigoplus_{i=1}^pM$ and $\overline{M}$ denotes complex conjugation.
\end{restatable}

Note that the operators in the image and domain of the map may act in Hilbert spaces of different dimension ($n$ and $m$).
Initially no restriction is placed on this.
But from the form of the map it is manifest that $m = (p+q)n$ where $(p+q)\geq 1$ so as expected the simulator or dual system is at least as large as the target.

As a consequence of achieving a full mathematical characterisation, it is relatively straightforward to then show that other physical properties are preserved by encodings,
\begin{restatable}[\cite{Cubitt2019} Prop.~28 and discussion]{prop}{proptobyextra}\label{prop tobyextra}
An \emph{encoding} preserves additional physical properties such that there are relationships between:
\begin{enumerate}
\item Partition functions, $\trace\left(e^{-\beta\mathcal{E}(H)} \right) = (p+q)\trace \left(e^{-\beta H} \right)$;
\item Entropies, $S(\mathcal{E}(\rho)) = S(\rho) + \log(p+q)$.
\end{enumerate}
\end{restatable}
\noindent Therefore encodings satisfy condition V without explicitly demanding this as an axiom.
There is also a relationship between the entropies of a state and its encoded form which we highlight in \cref{prop tobyextra}.

Preserving the eigenspectra of Hermitian operators hints towards preserving measurement outcomes.
However, conditions IV and VI additionally require a corresponding map on states to be well defined.
While \cite{Cubitt2019} provide examples of maps on states that, when considered with encodings, give conditions IV and VI, the form of $\mathcal{E}_\textup{state}$ is not characterised.
Note that while the eigenspectra is preserved, the eigenstates of operators including the Hamiltonian may look completely different in the original and encoded case, due to the unitary transformation allowed.
However, in particular constructive examples of simulations a close connection between eigenstates can be established see e.g.~Lemma~20 of \cite{ApelBaspin} and discussion therein.
\cref{Map on states} in fact demonstrates that the form of the map on states is also characterised as an implication of the definition of duality maps on observables and preserving measurement outcomes whereby the state mapping uncomputes this unitary transformation.


An earlier work also posed a definition of simulation based on an isometric encoding map \cite{Bravyi2014}.
\cite{Cubitt2019} includes more general maps than simple isometries since anything that satisfies the conditions in \cref{thm Encodings} are allowed.
\cite{Cubitt2019} also largely restricts to local encodings as the physically relevant case, whereas \cite{Bravyi2014} imposes no formal conditions on the isometry except noting it should be able to be implemented practically.

This framework was altered to consider a simulator system that only reproduces the ground state  and first excited state (and hence the spectral gap) of the Hamiltonian, in \cite{Aharonov2018}.
The independent interest of gap simulation is demonstrated by applying the framework to the task of Hamiltonian sparsification -- exploring the resources required for simplifying the Hamiltonian interaction graph.
Aside from the above works there has been little other follow up work exploring the theoretical notion of analogue Hamiltonian simulation and duality.

\subsection{Motivating examples}\label{sect example}

The framework analysed in \cite{Cubitt2019}, while the strongest sense of simulation/duality, already encompasses some important cases of physical dualities.
For example, fermionic encodings such as the Jordan-Wigner transformation \cite{JordanWigner, Nielsen2005} fit into the framework, able to replicate the full physics of the target in the simulator system.
Quantum error correcting codes are also examples of `simulations in a subspace' characterised by \cite{Cubitt2019}.
In this vein, to begin generalising the current literature we look to physical dualities not yet contained by the current frameworks.

As discussed above, when motivated by duality as opposed to just simulation, an important class is strong-weak dualities.
For many of these dualities, there are several aspects that prevent integration with the current frameworks.
These challenges can arise due to the absence of a comprehensive mathematical description (such as in the AdS/CFT duality) or the reliance on descriptions involving infinite-dimensional field theories.
However, there are simpler instances that still capture some characteristics of strong-dualities while being describable on a finite spin lattice.
Here we describe two examples of strong weak duality that are closest to our setting which we will use as motivation when extending the current framework.

\paragraph{Kramer-Wannier duality}

A paradigmatic example of a strong-weak duality is the Kramer-Wannier duality~\cite{Baxter1989}.
Even the isotropic case of this classical duality is not captured by the strong sense of simulation in~\cite{Cubitt2019} with the key novel element being the strong-weak nature of the two Hamiltonians.
Therefore this duality was a first benchmark for this generalisation of the theory of simulation to more broadly encompass dualities.

In Kramer-Wannier an Ising Hamiltonian on a 2d square lattice at high temperature ($\tanh J \beta \ll1$):
\begin{equation}
H = - J \sum_{\langle i,j \rangle} \sigma_i \sigma_j,
\end{equation}
is dual to another Ising Hamiltonian on the same lattice (in the thermodynamic limit) at low temperature ($\tilde{J}\tilde{\beta}\gg1$):
\begin{equation}
\Phi(H) = -\tilde{J} \sum_{\langle i,j \rangle} \sigma_i \sigma_j,
\end{equation}
in the thermodynamic limit.
The two Hamiltonians are dual, in the sense that their free energies, $f$, are related by
\begin{equation}
\tilde{\beta}f_{\Phi(H)}= \beta f_H + \ln \sinh (2\beta J),
\end{equation}
when the following duality condition relating the interaction strengths and temperature is satisfied:
\begin{equation}\label{condition eqn}
\tilde{J}\tilde{\beta} = - \frac{1}{2}\ln \tanh (J \beta).
\end{equation}
A more detailed description of this duality and how it arises is given in \cref{appen K-W}.

This duality can be used to find the critical point for the 2d Ising model since at this point the free energies will be non-analytic.
It is in some sense a very simple duality as both Hamiltonians have the same form and act on identical copies of the Hilbert space.
However, it follows from the non-trivial nature of the relation between the free energies that expecting all observables to be preserved is too strong.
Furthermore it is clear from the form of the duality that the energy spectrum cannot be preserved without a rescaling.
These two aspects of the duality prevent it from fitting into the framework developed in \cite{Cubitt2019}.

\paragraph{Boson-Fermion duality}

Boson-Fermion dualities (bosonisation/fermionisation) are a class of dualities transforming between bosonic and fermionic systems, usually in the context of quantum fields.
They are an example of particle vortex dualities that have had wide application, particularly in quantum field theory and condensed matter physics.
Similarly to the Kramer-Wannier duality the interest often lies in transforming strongly interacting fermionic systems (e.g. electrons in metals in condensed matter physics) to weakly interacting bosonic systems or vice versa.
In particular these dualities often work well near critical points or phase transitions where the crossover of these regimes takes place.
In this context the `strong-weak' nature of the duality can be referred to as a `UV to IR' duality.

There has been extensive study of boson-fermion duality in different dimensions and it is conjectured that an exact duality exists in 3D on the level of partition functions \cite{Polyakov:1988md, Kachru2016,Mross}.
Here `exact' duality refers to a transformation that is valid in all regimes including at criticality, whereas `approximate' dualities can be demonstrated to hold under some conditions or in specific UV or IR limit.
Extending the mathematical framework to also include this type of approximate duality is considered in \cref{sect Approximate dualities} where the equivalence is restricted to a subspace e.g. the low energy subspace.
The majority of this paper considers exact mappings, and a duality of this type was demonstrated between 3D lattice gauge theories in \cite{Chen2017}.

The duality in \cite{Chen2017} is between a strongly coupled boson and its free fermion vortex.
The bosonic theory is an $XY$ model coupled to a $U(1)$ Chern-Simons gauge field, where the Chern-Simons theory is realised via a lattice fermion with mass $M$ and interaction $U$.
The fermionic dual theory is a free massless Dirac fermion implemented by a lattice fermion of mass $M'$ and interaction $U'$.
The partition function of the fermionic system is shown to be proportional to the bosonic theory even at criticality given that the mass and interactions of the two lattice fermions are related via,
\begin{equation}
\frac{M'}{M} = \frac{I_0(1/T)}{I_1(1/T)}=\sqrt{\frac{1+U'}{1+U}} = \begin{cases}
1 \quad \text{when } T=0,\\
\infty \quad \text{when } T=\infty
\end{cases}
\end{equation}
where $I_j(x)$ is the $j$th modified Bessel function.
The above echoes \cref{condition eqn} giving the duality condition relating the physics of the two systems in different regimes (low and high temperature).

Generally the literature on boson-fermion dualities is out of reach for a duality framework considering operators in finite dimensions as there is a notable gap in our understanding of connecting quantum field theories to finite-dimensional operator algebra.
Nevertheless, examining the qualitative aspects of boson-fermion dualities can shed light on deficiencies within the previous mathematical framework.
This example reinforces the importance of incorporating non-trivial relationships between spectra to adequately accommodate strong-weak dualities.

\section{Generalised duality map}\label{sect Generalised duality map}

The first step in studying maps between operators describing a ``duality'' is to identify what properties these maps should preserve in general.
There is potential for wide variation in how duality maps are defined.
This work aims for a minimal set of axioms that encompasses as many dualities as possible, in particular strong-weak and high-low temperature dualities, while capturing~\cite{Cubitt2019}'s simulation as a special case.
This paper is restricted to consider finite dimensional systems, we denote Hermitian $(n\times n)$ matrices by $\text{Herm}_n$.

\begin{restatable}[Measurement duality map]{defn}{axioms1}
\label{Measurement duality axioms}
A \emph{measurement duality map}, $\Phi_s: \textup{Herm}_n \mapsto\textup{Herm}_{m}$ satisfies
\begin{enumerate}[(i)]
\item $\forall$ $a_i\in \textup{Herm}_n$, $p_i\in [0,1]$ with $\sum_i p_i = 1:$\\
  \[\Phi_s\left(\sum_i p_i a_i \right) = G\left(\sum_i p_i a_i \right)\sum_i g( a_i)h(p_i) \Phi_s(a_i);\]
\item $\forall$ $A\in \textup{Herm}_n:$\\
\[\textup{spec}\left[\Phi_s(A)\right] = f(A) \textup{spec}[A].\]
\end{enumerate}
The scaling functions $f$, $G$, $g$: $\textup{Herm}_n\mapsto \mathbb{R}$, are Lipschitz on any compact subset of $\textup{Herm}_n$ and map to zero iff the input is the zero operator.
$h$: $[0,1] \mapsto [0,1]$ describes a mapping between probability distributions such that $\sum_i h(p_i) = 1$.
\end{restatable}

Intuitively, all duality maps must preserve Hermiticity for observables in one theory to be associated with observables in another -- this is the most straightforward condition on any duality map.
The map is defined to take $(n\times n)$ Hermitian matrices as inputs and output $(m\times m)$ Hermitian matrices.
A priori there is no constraint or relation between $n$ and $m$ but we will later see as a consequence of the definition that $m/n$ is a positive integer.

Dualities are also constrained by the convex structure of quantum mechanics, but formulating the minimal requirements in this case is more subtle.
Operationally, a convex combination of observables corresponds physically to the process of selecting an observable at random from some ensemble of observables according to some probability distribution, measuring that observable, and reporting the outcome.
This is commonly described mathematically by an ensemble of observables: $\{p_i,A_i\}$, where $p_i$ is the probability of measuring observable $A_i$.
Since this is a physical operation that can be performed on the original system, there must be a corresponding procedure on the dual system that gives the same outcome.
However, this does \emph{not} imply that the dual process must necessarily be given by the convex combination of the dual observables.
It would clearly be possible operationally to first rescale the probability distribution before picking the dual observable to measure, and then to rescale the outcome of that measurement in some way before reporting it.
A fully general axiomatisation of duality has to allow for this possibility, and this is precisely what is captured mathematically in Axiom~(i).\footnote{Note that Axiom~(ii) is a slight abuse of notation since the map $\Phi$ is really a function of the ensemble $\{p_i,a_i \}$. However the outcome should not depend on how you chose to construct the ensemble average. It will turn out later (see \cref{lm Constrained scale functions} for details) that consistency with the final axiom imposes additional constrains the allowed probability and observable rescaling functions, such that $\Phi$ is truly only a function of the ensemble average. But this is a non-trivial consequence of the \emph{iteraction} between convexity and preservation of other physical properties; it is not required just by the duality of observable ensembles.}

In quantum mechanics measurement outcomes are associated with the spectra of the Hermitian operators, hence the final axiom requires a relation between the spectra of dual operators.
Again, operationally, we have to allow for the possibility of rescaling the measurement outcomes.
Even a simple change of measurement units, which has no \emph{physical} content, induces such a rescaling mathematically.
But more general rescalings that interchange large and small eigenvalues are possible, indeed required to encompass strong-weak dualities (e.g.\ the classic Kramer-Wannier duality).

This is captured mathematically in Axiom (ii) of~\cref{Measurement duality axioms} by the scaling function, $f$, which is observable-dependent.
Furthermore, Axiom~(ii) imposes a relation on the \emph{set} eigenvalues, but not on their ordering or multiplicities.
Thus which particular dual measurement outcome corresponds to which outcome on the original system can vary.
Since the scaling functions depend on the operator, the form of the duality is free to vary for different observables.

The only constraints imposed on the scaling functions $f,g,G$ are those we argue are physically necessary: the range must be restricted to real numbers since all measurement outcomes in quantum mechanics must be real; they are required to satisfy a very weak Lipschitz condition to exclude unphysical discontinuities; and non-vanishing for a non-zero input ensures every observable has a corresponding dual.

There are still plausible notions of duality not captured by this definition.
However, the formulation given in~\cref{Measurement duality axioms} is sufficient to restrict to mappings that represent meaningful dualities, yet be a substantial generalisation of \cref{thm Encodings}.

\subsection{Characterisation}\label{Characterisation}

A priori, requiring that the spectrum of operators is preserved up to a function that is allowed to depend on the operator itself would appear to be an extremely weak constraint on the map.
For example, this function may arbitrarily rescale or invert the spectrum for different operators.
However, the interplay between spectrum rescaling and (rescaled) convexity introduces significantly more rigidity into the maps' structure than either constraint alone.
The scaling functions $f,g,G$ appearing in the axioms are found to be necessarily related, such that the axioms can be equivalently rewritten using only a single function.
These relationships are proven rather than assumed by initially considering the action of the duality map on orthogonal projectors and proving that the constraints imply a non-trivial preservation of orthogonality (and then building up to general Hermitian operators).
A key element of this proof is that intuitively unphysical actions of the map -- for example discontinuous permutations within projectors during continuous variations in the operator -- can be ruled out using the analyticity of the resolvent of operators at non-degenerate points in its spectrum \cref{appen matching}.

\begin{restatable}[Characterisation]{thm}{thmcharacterisation}
\label{thm Characterisation}
Any \emph{measurement duality map}, $\Phi_s$, with the scale function $f(\cdot)$ is necessarily of the form,
\[\Phi_s(A) = f(A) U \left(A^{\oplus p} \oplus \overline{A}^{\oplus q} \right) U^\dagger,\]
where $p,q$ are non-negative integers, $U$ is a unitary transformation and $\bar{A}$ represents the complex conjugate of $A$.
Equivalently,
\[\Phi_s(A) = f(A) U \left(A\otimes P + \overline{A}\otimes Q \right) U^\dagger,\]
where $P$ and $Q$ are orthogonal complemently projectors.
\end{restatable}

\begin{figure}[tbp]
\centering
\begin{tikzpicture}
\begin{scope}[on background layer]
         \node [inner sep=0pt] (1) at (0, 0) {\includegraphics[trim={0cm -1cm 0cm 0cm},clip,scale=0.43]{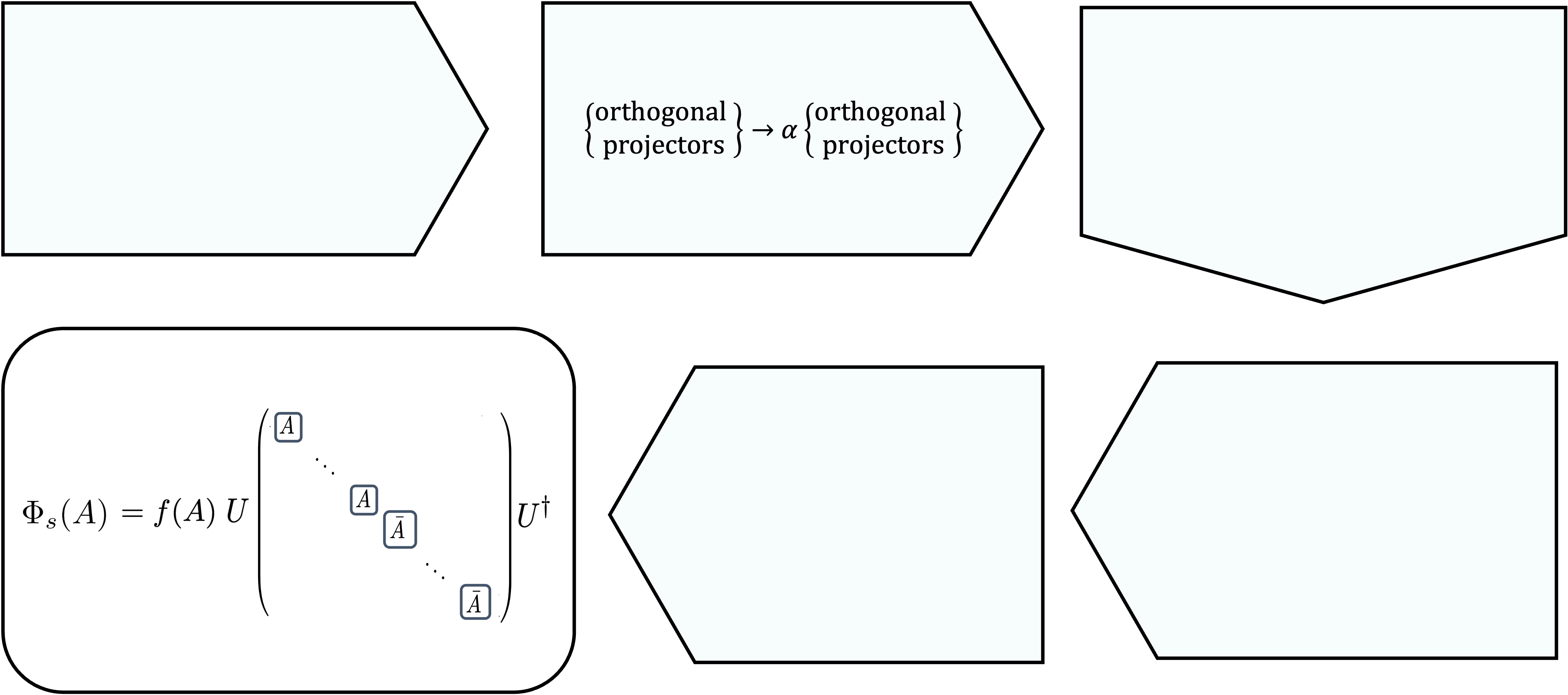}};
\end{scope}
		\node [align=left, anchor=west] (0) at (-6.5, 2.3) {\scriptsize Measurement duality map,};
		\node [align=left, anchor=west] (0) at (-6.5, 1.9) {\scriptsize  $\Phi_s$, \cref{Measurement duality axioms}};
		\node [align=left, anchor=west] (0) at (-2, 2.7) {\scriptsize $\Phi_s$ maps,};
		\node [align=left, anchor=west] (0) at (-2, 1.5) {\scriptsize (\cref{lm Mapping orthogonal complement projectors}, \cref{lm Mapping orthogonal projectors})};
		\node [align=left, anchor=west] (0) at (2.8, 2.3) {\scriptsize Consider general Hermitians};
		\node [align=left, anchor=west] (0) at (3, 1.9) {\scriptsize in spectral decomposition};
		\node [align=left, anchor=west] (0) at (2.8, -0.8) {\scriptsize A relationship between scale };
		\node [align=left, anchor=west] (0) at (2.8, -1.2) {\scriptsize functions in the definition of };
		\node [align=left, anchor=west] (0) at (2.8, -1.6) {\scriptsize$\Phi_s$ (\cref{lm Constrained scale functions})};
		\node [align=left, anchor=west] (0) at (-1.0, -0.8) {\scriptsize Connect to C$*$ algebra };
		\node [align=left, anchor=west] (0) at (-1.0, -1.2) {\scriptsize techniques seen in};
		\node [align=left, anchor=west] (0) at (-1.0, -1.6) {\scriptsize \cite{Cubitt2019}};
		\node [align=left, anchor=west] (0) at (-6.6, 0) {\scriptsize \cref{thm Characterisation}: map characterisation};
		\node [align=left, anchor=west] (0) at (-6.6, -2.4) {\scriptsize where $U$ is a unitary.};
\end{tikzpicture}
\caption{Outline of proof idea for \cref{thm Characterisation}. Starting with the definition of the duality map, we first show that orthogonal projectors are mapped to orthogonal projectors. This then allows one to consider Hamiltonians in their spectral decomposition and map these using the rescaled convexity axiom, there are two technical results as we first show this result for orthogonal complement projectors (\cref{lm Mapping orthogonal complement projectors}) and then general orthogonal projectors (\cref{lm Mapping orthogonal projectors}). From consistency this results in a relationship between the scale functions used to define the map. Substituting these relations into the definition gives new conditions in terms of a single rescaling function $f$. From here, the connection to the characterisation theorem in \cite{Cubitt2019} can be made which leads to the final result: the duality map is necessarily of the form of taking direct sum of copies of the observable and the complex conjugate of the observable, doing a unitary transformation and multiplying by the rescaling function.}
 \label{fg proof idea}
\end{figure}

The rest of this section is dedicated to proving \cref{thm Characterisation}.
A sketch of the argument and ingredients used in the proof are outlined in \cref{fg proof idea}.
The result relies on relating duality maps to the encodings characterised in \cref{thm Encodings}.
To demonstrate this, we first need to examine the necessary relations between the different scaling functions which in tern requires establishing how the map transforms orthogonal projectors.
The following lemma shows that a duality map will take orthogonal complement projectors to objects proportional to two new orthogonal complement projectors in the new Hilbert space.

\begin{restatable}[Mapping orthogonal complement projectors]{lemma}{lemmaMappingorthogonacomplementprojectors}
\label{lm Mapping orthogonal complement projectors}
Let $Q_1$ and $Q_2$ be orthogonal complement projectors ($Q_1Q_2 = Q_2Q_1 = 0$ and $Q_1 + Q_2 = \mathbb{I}$). Under a measurement duality map $\Phi_s$ these projectors are mapped to:
\[\Phi_s(cQ_1)\propto \Sigma_1 \qquad \Phi_s(cQ_2)\propto \Sigma_2.\]
Where $c\in\mathbb{R}$ and $\Sigma_1, \Sigma_2$ are themselves orthogonal complement projectors, i.e. $\Sigma_1^2 = \Sigma_1$, $\Sigma_2^2=\Sigma_2$, $\Sigma_1 \Sigma_2 = \Sigma_2 \Sigma_1 = 0$ and $\Sigma_1 + \Sigma_2 = \mathbb{I}$.
\end{restatable}

\begin{proof}
Since a general projector $P_i$ has $\text{spec}[P_i]\in \{0,1\}$, by axiom (ii) of \cref{Measurement duality axioms} the mapped operator has $\text{spec}\left[\Phi_s(cP_i) \right] = f(cP_i) \text{spec}[cP_i] = cf(cP_i)\text{spec}[P_i]\in  cf(cP_i)\{ 0,1 \}$.
The map also preserves Hermiticity via definition, so projectors are mapped to operators proportional to projectors.
In particular, given orthogonal complement projectors:
\begin{align}
&\Phi_s(cQ_1) = cf(cQ_1) \Sigma_1\\
&\Phi_s(cQ_2) = cf(cQ_2) \Sigma_2,
\end{align}
it only remains to show that $\Sigma_1,\Sigma_2$ are also orthogonal complement projectors.

The identity is a special case since $\text{spec}[\mathbb{I}] \in \{1 \}$ so $\text{spec}[\Phi_s(\mathbb{I})]\in \{ f(\mathbb{I})\}$. Therefore,
\begin{equation}
\Phi_s\left(\frac{c}{2}(Q_1 + Q_2)\right) = \Phi_s\left(c\mathbb{I}/2\right) =\frac{c}{2} f\left(c\mathbb{I}/2\right)\mathbb{I} \label{id_map}.
\end{equation}
Applying axiom (i) to the sum of operators gives,
\begin{align}
\Phi_s\left(\frac{c}{2}(Q_1 + Q_2)\right) &= G(c\mathbb{I}/2) h(1/2)\left[ g(cQ_1)\Phi_s(cQ_1) + g(cQ_2)\Phi_s(cQ_2) \right] \\
& = G(c\mathbb{I}/2)h(1/2)\left[g(cQ_1)cf(cQ_1)\Sigma_1 + g(Q_2)cf(cQ_2)\Sigma_2 \right] \label{sum_axiom_result}.
\end{align}
Note that while $c$ is a general real, (i) has to be applied with $\sum_i p_i = 1$ and $p_i \in [0,1]$, in this case $t_1,t_2 = 1/2$ and $c$ has been absorbed into the Hermitian operators.

Equating \cref{id_map} and \cref{sum_axiom_result},
\begin{align}
G(c\mathbb{I}/2)h(1/2) \left[g(Q_1)cf(cQ_1)\Sigma_1 + g(Q_2)cf(cQ_2)\Sigma_2 \right] =& \frac{c}{2}f(c\mathbb{I}/2) \mathbb{I}\\
2\frac{G(c\mathbb{I}/2)}{f(c\mathbb{I}/2)}h(1/2)\left[g(cQ_1)f(cQ_1) \Sigma_1 + g(cQ_2)f(cQ_2)\Sigma_2 \right] =& \mathbb{I}\\
\alpha \Sigma_1 + \beta \Sigma_2 =& \mathbb{I},\label{sum of comp proj}
\end{align}
where the notation is simplified by defining:
\begin{equation}
\alpha := \frac{2G(c\mathbb{I}/2)h(1/2)g(cQ_1)f(cQ_1)}{f(c\mathbb{I}/2)}, \qquad
\beta := \frac{2G(c\mathbb{I}/2)h(1/2)g(cQ_2)f(cQ_2)}{f(c\mathbb{I}/2)} .
\end{equation}
Rewriting the matrices in \cref{sum of comp proj} in the $\{\Sigma_1, \Sigma_1^\perp\}$ basis,
\begin{equation}\label{eqn orthog comp key}
\alpha \left(
\begin{array}{c|c}
\mathbb{I} & 0 \\
\hline
 0 & 0
\end{array}
\right) +
\beta \left(
\begin{array}{c|c}
A & B \\
\hline
C  & D
\end{array}
\right) =
\left(
\begin{array}{c|c}
\mathbb{I} & 0  \\
\hline
0  & \mathbb{I}
\end{array}
\right).
\end{equation}
Equating the off-diagonal quadrants gives that $ \beta B=\beta C = 0$.
Since the initial properties of the scaling functions imply that $\beta\neq 0$, $B$ and $C$ must vanish and $\Sigma_1, \Sigma_2$ are simultaneously diagonalisable with $\left[\Sigma_1, \Sigma_2\right] =0$.
Equating diagonal quadrants gives:
\begin{align}
& \alpha \mathbb{I} + \beta A = \mathbb{I}\label{eqn for A}\\
& \beta D = \mathbb{I}.\label{eqn for D}
\end{align}
In order for $\Sigma_2$ to be a valid projector $D^2=D$ and $A^2=A$.
This together with the expression for $D= \frac{1}{\beta} \mathbb{I}$ from \cref{eqn for D} implies that $\beta = +1$ and $D=\mathbb{I}$.
Finally, rearranging \cref{eqn for A},
\begin{align}
A = (1-\alpha)\mathbb{I} = A^2 = (1-\alpha)^2\mathbb{I},
\end{align}
together with $\alpha\neq0$ implies that $\alpha =+1$, $A=0$. In both the above cases, the solutions $\beta=-1$ and $\alpha=-2$ are discarded since $\Sigma_2$ must be a positive definite operator.
In the $\{\Sigma_1, \Sigma_1^\perp\}$ basis
\begin{equation}
\Sigma_1 = \left(
\begin{array}{c|c}
\mathbb{I} & 0 \\
\hline
0 & 0
\end{array}
\right) \qquad
\Sigma_2 = \left(
\begin{array}{c|c}
0 & 0 \\
\hline
0 & \mathbb{I}
\end{array}
\right),
\end{equation}
so $\Sigma_1,\Sigma_2$ are orthogonal complement projectors.
\end{proof}

The expressions for $\alpha$ and $\beta$ give some initial relations between the scale functions appearing in the axioms:
\begin{equation}
h(1/2)g(Q_1)f(cQ_1) = h(1/2)g(cQ_2)f(cQ_2) = \frac{f(c\mathbb{I}/2)}{2G(c\mathbb{I}/2)}. \label{eqn proj fg}
\end{equation}
Since for any projector $P_i$ there exists its complement $P_i^\perp$, it follows that the above applies generally for any projector: $h(1/2)g(cP_i)f(cP_i) =\frac{f(c\mathbb{I}/2)}{2G(c\mathbb{I}/2)}$.

Now a statement concerning how a measurement duality map acts on two orthogonal projectors that only span a subspace of the initial Hilbert space can be made.

\begin{restatable}[Mapping orthogonal projectors]{lemma}{lemmaMappingorthogonalprojectors}
\label{lm Mapping orthogonal projectors}
Let $P_1$ and $P_2$ be orthogonal projectors such that $P_1P_2 = P_2P_1 = 0$.
Under a measurement duality map, $\Phi_s$, these projectors are mapped to:
\[\Phi_s(cP_1)\propto \Pi_1 \qquad \Phi_s(cP_2)\propto \Pi_2,\]
where $c\in \mathbb{R}$ and $\Pi_1, \Pi_2$ are themselves orthogonal projectors.
\end{restatable}

\begin{proof}
Again spectrum preservation stipulates that projectors are mapped to objects proportional to projectors:
\begin{align}
&\Phi_s(cP_1) = cf(cP_1)\Pi_1\\
&\Phi_s(cP_2) = cf(cP_2) \Pi_2 \\
&\Phi_s\left(\frac{c}{2}(P_1+P_2)\right) = \frac{c}{2}f\left(\frac{c}{2}(P_1 + P_2)\right) \Pi_{12}, \label{eqn pi12 def}
\end{align}
where the final equation holds since the sum of two orthogonal projectors is another projector.
Applying axiom (i) to the sum and substituting the above:
\begin{align}
\Phi_s\left(\frac{c}{2}(P_1 + P_2)\right) & \begin{multlined}
= G\left(\frac{c}{2}(P_1 + P_2)\right) h(1/2)\left[g(cP_1)\Phi_s(cP_1) \right.\\
\left.+ g(cP_2)\Phi_s(cP_2) \right]
\end{multlined}\\
&\begin{multlined}
 = G\left(\frac{c}{2}(P_1 + P_2)\right) h(1/2)\left[g(cP_1)cf(cP_1) \Pi_1 \right.\\
\left.+ g(cP_2)cf(cP_2)\Pi_2 \right]. \label{eqn axiom 2 orthog proj}
\end{multlined}
\end{align}
Equating \cref{eqn pi12 def} and \cref{eqn axiom 2 orthog proj} in the same way as in \cref{lm Mapping orthogonal complement projectors} gives:
\begin{equation}\label{alpha_equ}
\alpha (\Pi_1 + \Pi_2 )= \Pi_{12}
\end{equation}
where
\begin{align}
\alpha &= \frac{2G(c/2(P_1 +P_2))h(1/2)g(cP_1)f(cP_1)}{f(c/2(P_1 + P_2))} \\
&= \frac{2G(1/2(P_1 + P_2)) h(1/2)g(cP_2)f(cP_2)}{f(c/2(P_1 + P_2))}\\
& = \frac{G(c/2(P_1+P_2))}{f(c/2(P_1+P_2))}\frac{f(c\mathbb{I}/2)}{G(c\mathbb{I}/2)}.
\end{align}
In the above, the scale factor relation for projectors from \cref{eqn proj fg} is used to equate $g(cP_1)f(cP_1)=g(cP_2)f(cP_2)$.

Writing the matrices in \cref{alpha_equ} in the $\{\Pi_{12}, \Pi_{12}^\perp \}$ basis:
\begin{equation}
\alpha \left[ \left(
\begin{array}{c|c}
A_1 & B \\
\hline
C & D
\end{array}
\right) + \left(
\begin{array}{c|c}
A_2 & -B \\
\hline
-C & -D
\end{array}
\right)\right] = \left(
\begin{array}{c|c}
\mathbb{I} & 0 \\
\hline
0 & 0
\end{array}
\right).
\end{equation}
Since $\Pi_1,\Pi_2$ are projectors, they must be positive semi-definite matrices.
Let $\ket{x}$ be a vector only with support on the $\Pi_{12}^\perp$ subspace.
The positive semi-definite property requires that
\begin{align}
&\bra{x}\Pi_1 \ket{x}  = \left(
\begin{array}{cc}
0 & x
\end{array}
\right) \left(
\begin{array}{c|c}
A_1 & B \\
\hline
C & D
\end{array}
\right) \left(
\begin{array}{c}
0  \\
x
\end{array}
\right) = Dx^2 \geq 0 \\
&\bra{x}\Pi_2 \ket{x}  = \left(
\begin{array}{cc}
0 & x
\end{array}
\right) \left(
\begin{array}{c|c}
A_2 & -B \\
\hline
-C & -D
\end{array}
\right) \left(
\begin{array}{c}
0  \\
x
\end{array}
\right) = -Dx^2 \geq 0.
\end{align}
Only $D=0$ can satisfy the above simultaneously.
Once the lower right block is set to 0, the off-diagonal blocks must also vanish for $\Pi_i$ to be valid projectors (see \cref{Vanishing off-diagonal matrix elements}),
\begin{equation}
\Pi_1 = \left(
\begin{array}{c|c}
A_1 & 0 \\
\hline
0 & 0
\end{array}
\right), \qquad
\Pi_2 =\left(
\begin{array}{c|c}
A_2 & 0 \\
\hline
0 & 0
\end{array}
\right) .
\end{equation}
Therefore \cref{alpha_equ} reduces to the same form as \cref{eqn orthog comp key} when examining the top left quadrant only,
\begin{equation}
\alpha (A_1 + A_2) = \mathbb{I},
\end{equation}
identifying that $\alpha=1$ since $A_1,A_2$ are projectors.
Applying \cref{lm Mapping orthogonal complement projectors} gives $A_1A_2 = A_2A_1 = 0$.
The result is that $\Pi_1\Pi_2 = \Pi_2\Pi_1 =0$.
\end{proof}

A consequence of $\alpha=1$ is that,
\begin{equation}
\frac{G(c/2(P_1+P_2))}{f(c/2(P_1+P_2))} = \frac{G(c\mathbb{I}/2)}{f(c\mathbb{I}/2)},
\end{equation}
for all orthogonal projectors $P_1,P_2$.
The above relation can be shown to hold in a more general case which leads to a restatement of the axiom describing the behaviour of the map acting on convex combinations.

\begin{restatable}[Constrained scale functions]{lemma}{lemmaconstrainedscalefunctions}
\label{lm Constrained scale functions}
A duality map, $\Phi_s$, satisfies
\begin{enumerate}[(i')]
\setcounter{enumi}{0}
\item $\Phi_s(\sum_i p_i a_i) = f(\sum_i p_i a_i)\sum_i \frac{p_i}{f(a_i)}\Phi_s(a_i)$
\end{enumerate}
for all $a_i\in\textup{Herm}_n$ and $p_i\in[0,1]$ with $\sum_i p_i = 1$.
\end{restatable}

\begin{proof}
This proof follows by demonstrating various relationships between the scaling functions $f,g,G$ that must hold as a consequence of \cref{Measurement duality axioms}.

First, for all Hermitian operators $A$, the ratio of $f(A)$ to $G(A)$ is proven to be a constant independent of $A$.
The spectral decomposition of a general Hermitian operator $A$ is given by
\begin{equation}
A = \sum_i \lambda_i P_i,
\end{equation}
where $\lambda_i\in \mathbb{R}$ and in the case of degenerate eigenvalues we are free to chose $\{ P_i\}$ to form a set of orthogonal projectors.
In order to apply axiom~(i) of \cref{Measurement duality axioms} the summation is rearranged to read,
\begin{equation}
A = \sum_i \mu_i \left(c_i P_i \right),
\end{equation}
where now $\mu_i \in [0,1]$ and $\sum_i \mu_i = 1$, whereas $c_i \in \mathbb{R}$ with $\mu_i c_i = \lambda_i$.
Note that while clearly this choice of $\mu_i c_i $ is not unique, this does not affect the following argument.

By axiom~(i) of \cref{Measurement duality axioms},
\begin{equation}\label{eqn ax 2 general herm}
\Phi_s(A) = G(A) \sum_i h(\mu_i) g(c_iP_i)\Phi_s(c_i P_i).
\end{equation}
Using \cref{lm Mapping orthogonal projectors} this can be written as a spectral decomposition over orthogonal projectors,
\begin{equation}
\Phi_s(A) = G(A) \sum_i h(\mu_i) g(c_iP_i) f(c_iP_i)c_i \Pi_i.
\end{equation}
However since the spectral decomposition is unique (up to degenerate eigenvalues where we continue to chose an orthogonal basis) it can also be expressed using the spectrum preserving axiom as
\begin{equation}\label{eqn ax 3 general herm}
\Phi_s(A) = f(A) \sum_i \mu_ic_i \Pi_{\sigma(i)},
\end{equation}
where $\sigma(i)$ denotes some permutation of indices.

Equating \cref{eqn ax 2 general herm} and \cref{eqn ax 3 general herm} gives,
\begin{equation}
\frac{G(A)}{f(A)} \sum_i h(\mu_i) g(c_iP_i) f(c_iP_i)c_i \Pi_i= \sum_i \mu_ic_i \Pi_{\sigma(i)}.
\end{equation}
Multiplying by $\Pi_k$ selects for a given projector,
\begin{equation}
\frac{G(A)}{f(A)}  h(\mu_{\sigma(j)}) g(c_{\sigma(j)}P_{\sigma(j)}) f(c_{\sigma(j)}P_{\sigma(j)})c_{\sigma(j)} = \mu_jc_j ,
\end{equation}
where $\sigma(j)=k$.
\cref{appen matching} demonstrates that in fact $\sigma(k)=k$ $\forall k$ is the only allowed permutation for any map $\Phi_s$ and operator $A$.
Therefore we can equate
\begin{equation}
h(\mu_i) g( c_iP_i) = \frac{f(A)}{G(A)} \frac{\mu_i}{f(c_iP_i)}.
\end{equation}

Since $h(\mu_i) g( c_iP_i)$ cannot depend on the other eigenvalues and vectors of $A$ the ratio of $f(A)$ to $G(A)$ must be constant for any given Hermitian, i.e.
\begin{equation}\label{eqn f and G prop}
\frac{f(A)}{G(A)} = x, \qquad \forall A \in \text{Herm},
\end{equation}
for some $x\in \mathbb{R}$.

Applying (i) of \cref{Measurement duality axioms} to the trivial sum ${\Phi_s(A) = G(A)h(t)g(A)\Phi_s(A)}$ gives another useful relation,
\begin{equation}\label{eqn g relationship}
g(A) = \frac{1}{G(A)} = \frac{x}{f(A)}, \qquad \forall A \in \text{Herm},
\end{equation}
since $h(1)=1$ by definition.

The next step is to investigate the function $h$ by relating $h(t)g(A)$ and $g(A)$.
Let $A_1$, $A_2$ be any two Hermitian operators with spectral decompositions,
\begin{align}
A_1 & = \sum_i \lambda_i P_i\\
A_2 & = \sum_i \mu_i Q_i,
\end{align}
where $\lambda_i,\mu_i \in \mathbb{R}$ such that $\{P_i,Q_i\}$ form an orthogonal set of projectors, i.e. $A_1$ and $A_2$ must have orthogonal support.
Consider a convex combination,
\begin{equation}
A = t A_1 + (1-t)A_2,
\end{equation}
with $t \in [0,1]$.
Since $A_1$ and $A_2$ have orthogonal support and the map obeys axiom~(ii) of \cref{Measurement duality axioms}, the spectrum of the mapped convex combination is:
\begin{equation}
  \text{spec}\left[\Phi_s(A) \right]=f(A)\{t\lambda_i, (1-t)\mu_i \}.
\end{equation}

On the other hand, applying axiom~(i) of \cref{Measurement duality axioms} to $A$ gives,
\begin{equation}
\Phi_s(A) = G(A)\left[h(t) g(A_1)\Phi_s(A_1) + h(1-t) g(A_2)\Phi_s(A_2) \right].
\end{equation}
By \cref{lm Mapping orthogonal projectors}, $\Phi_s(A_1)$ and $\Phi_s(A_2)$ have orthogonal support, and $\{\Phi_s(\lambda_i P_i),\Phi_s(\mu_iQ_i) \}$ is an orthogonal set.
Together with axiom~(ii) of \cref{Measurement duality axioms}, this implies that
\begin{align}
\text{spec}[\Phi_s(A)]
  &= \{G(A)h(t)g(A_1)\text{spec}[\Phi_s(A_1)],\; G(A)h(1-t)g(A_2)\text{spec}[\Phi_s(A_1)]\}\\
  &= \{G(A)h(t) g(A_1)f(A_1)\lambda_i,\; G(A)h(1-t)g(A_2)f(A_2)\mu_i\}.
\end{align}

Again using the result from \cref{appen matching} that the permutation is trivial, we can equate the elements of $\text{spec}\left[\Phi_s(A) \right]$ that correspond to $A_1$:
\begin{equation}
f(A)t \lambda_i = G(A)h(t)g(A_1)f(A_1)\lambda_i.
\end{equation}

Using \cref{eqn f and G prop} and \cref{eqn g relationship},
\begin{equation}\label{eqn g relationship 2}
h(t)g(A_1) = \frac{t x}{f(A_1)},
\end{equation}
for all $ A_1\in \text{Herm}$ and $ t\in[0,1]$.

Finally, substituting for $g,G$ using \cref{eqn f and G prop} and \cref{eqn g relationship 2}, (ii) of \cref{Measurement duality axioms} becomes,
\begin{align}
\Phi_s\left(\sum_i p_i A_i \right) & = G\left(\sum_i p_i A_i \right) \sum_i h(p_i)g(A_i)\Phi_s(A_i)\\
& = \frac{f\left(\sum_i p_i A_i \right)}{x}\sum_i \frac{p_i x}{f(A_i)}\Phi_s(A_i)\\
& = f\left(\sum_i p_i A_i \right) \sum_i \frac{p_i}{f(A_i)}\Phi_s(A_i)
\end{align}
for all $A_i\in\text{Herm}$ and $p_i \in [0,1]$ where $\sum_i p_i = 1$.
\end{proof}

This constraint on how the map acts on convex combinations of operators enables the link between duality maps and the encodings in \cref{thm Encodings} to be made.

\begin{proof} (of \cref{thm Characterisation})
To characterise $\Phi_s$ we define the related map $\mathcal{E}(A): = \frac{\Phi_s(A)}{f(A)}$ and show that $\mathcal{E}$ is an encoding in the sense of \cref{thm Encodings}.
For $\mathcal{E}$ to be an encoding it is sufficient to show that is satisfies the 3 conditions given in \cref{thm Encodings}.

\cref{Measurement duality axioms} states $\Phi_s(A)^\dagger=\Phi_s(A)$, therefore
\begin{equation}
\mathcal{E}(A)^\dagger = \frac{\Phi_s(A)^\dagger}{\overline{f(A)}} = \frac{\Phi_s(A)}{\overline{f(A)}}.
\end{equation}
However, $\overline{f(A)} = f(A)$ since it is defined be a real function.
Therefore $\mathcal{E}(A)^\dagger = \mathcal{E}(A)$ and the first encoding axiom is satisfied.

Using (ii) of \cref{Measurement duality axioms}, it quickly follows that $\mathcal{E}$ is spectrum preserving:
\begin{align}
\text{spec}\left[ \mathcal{E}(A)\right] & = \text{spec} \left[ \frac{\Phi_s(A)}{f(A)} \right]\\
& = \frac{1}{f(A)}\text{spec}\left[ \Phi_s(A)\right]\\
& = \frac{1}{f(A)} f(A)\text{spec}[A]\\
& = \text{spec}[A].
\end{align}

The final encoding axiom is shown using (i) of \cref{Measurement duality axioms} and \cref{lm Constrained scale functions} to demonstrate that $\mathcal{E}$ is convex,
\begin{align}
\mathcal{E}(\sum_i p_i a_i) & = \frac{\Phi_s(\sum_i p_i a_i)}{f(\sum_i p_i a_i)}\\
&= \frac{1}{f(\sum_i p_i a_i)}f(\sum_i p_i a_i) \sum_i \frac{p_i}{f(a_i)}\Phi_s(a_i)\\
& = \sum_i \frac{p_i}{f(a_i)} f(a_i)\mathcal{E}(a_i)\\
&= \sum_i p_i \mathcal{E}(a_i).
\end{align}

The mathematical form follows directly from $\Phi_s(A)=f(A)\mathcal{E}(A)$ and \cref{thm Encodings}.
\end{proof}

\subsection{Map on states}\label{Map on states}

A map on Hamiltonians and observables is not enough to fully characterise the duality, since a state in one theory should also have a corresponding state in the other.
The set of states is just a subset of Hermitian operators, however the physical requirements on the state map differ to those given in \cref{Measurement duality axioms}.
Instead, when we consider maps on states, we need them to be compatible with the map on operators such that measurement outcomes and time dynamics behave as expected.
In the following definition we use $\mathcal{H}_n$ to denote a Hilbert space of dimension $(n\times n)$ and $\mathcal{S}(\mathcal{H})$ to denote the set of states in Hilbert space $\mathcal{H}$.

\begin{restatable}[Compatible duality state map]{defn}{defstatemap}\label{defn state map}
Given a duality map, $\Phi$, on operators (\cref{Measurement duality axioms}), we say that a map on states, $\Phi_\textup{state}: \mathcal{S}(\mathcal{H}_n)\mapsto \mathcal{S}(\mathcal{H}_m)$, is compatible with $\Phi$ if is satisfies the following properties:
  \begin{enumerate}
    \item convexity: for all ${p_i\in [0,1]}$ and ${\sum_i p_i = 1}$,
\[    {\Phi_\textup{state}(\sum_i p_i\rho_i) = \sum_i p_i\Phi_\textup{state}(\rho_i)};\]
  \item measurement outcomes are preserved up to the scaling function,
\[      \trace \left[\Phi(A)\Phi_\textup{state}(\rho) \right]   = f(A) \trace \left[A \rho \right]\]
    for all $A\in\textup{Herm}_n$, $\rho\in\mathcal{S}(\mathcal{H}_n)$;
  \item time dynamics is consistent at rescaled times,
\[      \Phi_\textup{state}\left(e^{-iHt}\rho e^{iHt} \right) = e^{-i\Phi(H)t/f(H)}  \Phi_\textup{state}(\rho) e^{i\Phi(H)t/f(H)}.\]
  \end{enumerate}
\end{restatable}

While examples of compatible maps on states were given for the simulations in \cite{Cubitt2019} this section proves that the form of the map on states is implied by the definitions of duality maps and the corresponding map on states.

\begin{restatable}[Form of state map]{prop}{propformofstatemap}
\label{prop form of state map}
Given a duality map, $\Phi(A) = f(A) U \left(\bigoplus_{i=1}^p A \oplus \bigoplus_{i=p+1}^{p+q} \bar{A}\right)U^\dagger$, on operators, the compatible duality map on states, $\Phi_\textup{state}: \mathcal{S}(\mathcal{H}_n)\mapsto \mathcal{S}(\mathcal{H}_m)$, as in~\cref{defn state map}, is necessarily of the form:
\[\Phi_\textup{state}(\rho) = U\left(\bigoplus_{i=1}^{p}\alpha_i \rho \right)U^\dagger,\]
where $\alpha_i \in [0,1]$ and $\sum_{i=1}^p \alpha_i= 1$.
\end{restatable}

\begin{proof}
Setting $B = e^{iHt}$ and conjugating condition 3 of compatible duality state maps with $U^\dagger$
\begin{align}
U^\dagger \Phi_\text{state}\left(B \rho B^\dagger \right) U&= U^\dagger e^{i\Phi(H)t/f(H)}\Phi_\text{state}(\rho)e^{-\Phi(H)t/f(H)}U\\
& =  \left(B^{\oplus p}\oplus \bar{B}^{\oplus q} \right)U^\dagger \Phi_\text{state}(\rho) U \left((B^\dagger)^{\oplus p}\oplus (\bar{B}^\dagger)^{\oplus q} \right).
\end{align}
Since $B$ represents time evolution for general $t$ and $H$, the above shows that the conjugated state map must have the same block diagonal structure as $\Phi$, i.e.
\begin{equation}\label{eqn block diag structure}
U^\dagger\Phi_\text{state}(\rho)U = \bigoplus_{i=1}^{p+q}X_i (\rho).
\end{equation}

We now substitute this structure of the state map into condition 2 of the definition of compatible state maps:
\begin{align}
\trace(A\rho) &= \trace \left[U \left(\bigoplus_{i=1}^p A \oplus \bigoplus_{i=p+1}^{p+q} \bar{A} \right)U^\dagger U \bigoplus_{i=1}^{p+q} X_i(\rho)U^\dagger \right] \\
& = \trace\left[ \bigoplus_{i=1}^p A X_i(\rho) \oplus \bigoplus_{i=p+1}^{p+q} \bar{A}X_i(\rho) \right]\\
& = \sum_{i=1}^{p}\trace\left[AX_i(\rho) \right] + \sum_{i=p+1}^q \trace \left[\bar{A}X_i(\rho) \right].\label{eqn to differentiate}
\end{align}
Since~\cref{eqn to differentiate} is true for all $A$ we can differentiate with respect to $A$,
\begin{equation}\label{eqn diff A}
\rho = \sum_{i=1}^p X_i(\rho),
\end{equation}
and separately with respect to $\bar{A}$,
\begin{equation}\label{eqn diff A comp}
0 = \sum_{i=p+1}^{p+q} X_i(\rho).
\end{equation}
Note that $A$ and $\bar{A}$ are independent for the purpose of differentiation.

The fact that $\Phi_\text{state}$ maps states to states implies that $X_i(\rho)$ is a positive operator for all $i$ and $\rho \in \mathcal{S}(\mathcal{H}_n)$.
Apply $X_i$ to a pure state $\ket{\psi_0}$ and assume for contradiction that the image has some support on a distinct pure state which wlog we call $\ket{\psi_1}$,
\begin{equation}
X_i\left(\ket{\psi_0}\bra{\psi_0} \right) = \alpha_i \ket{\psi_0}\bra{\psi_0} + \beta_i \ket{\psi_1}\bra{\psi_1} + \text{else},
\end{equation}
where "else" has no overlap with $\ket{\psi_0}$ or $\ket{\psi_1}$.
$0 \leq \alpha_i,\beta_i \leq 1$ since $X_i(\rho)$ is a positive operator.
From~\cref{eqn diff A},
\begin{align}
\ket{\psi_0}\bra{\psi_0} & = \sum_{i=1}^p X_i (\ket{\psi_0}\bra{\psi_0})\\
& = \sum_{i=1}^p \alpha_i \ket{\psi_0}\bra{\psi_0} + \beta_i \ket{\psi_1}\bra{\psi_1} + \text{else}.
\end{align}
Therefore $\sum_{i=1}^p \alpha_i = 1$ and $\sum_{i=1}^{p}\beta_i = 0$ $\implies$ $\beta_i = 0$ for all $i$.
Hence when applied to any pure state each $X_i$ for $i\in[1,p]$ acts as,
\begin{equation}\label{eqn X action on pure}
X_i(\ket{\psi}\bra{\psi}) = \alpha_i \ket{\psi}\bra{\psi} \qquad \text{with} \qquad \sum_{i=1}^p \alpha_i = 1.
\end{equation}

It follows from condition 1 that each $X_i$ is individually convex.
Explicitly
\begin{equation}
U \bigoplus_{i=1}^{p+q}X_i\left(\sum_j t_j \rho_j\right) U^\dagger= \sum_j t_j U \oplus_{i=1}^{p+q}X_i(\rho_j)U^\dagger
\end{equation}
implies that for all $i$ the following is true
\begin{equation}
X_i \left( \sum_j t_j \rho_j\right) = \sum_j t_j X_i (\rho_j).
\end{equation}
This combined with~\cref{eqn X action on pure} gives for any state $\rho\in\mathcal{S}(\mathcal{H}_n)$,
\begin{equation}\label{eqn for X}
X_i(\rho) = \alpha_i \rho \qquad \text{with} \qquad \sum_{i=1}^p \alpha_i = 1.
\end{equation}
By normalisation, $X_i(\rho) = 0$ for $i\in[p+1,q]$ which can also be seen from~\cref{eqn diff A comp} by applying a similar argument as for $b_i=0$.

\cref{eqn for X} combined with~\cref{eqn block diag structure} gives the quoted form of the map.
\end{proof}

\section{Equivalent definitions of duality}\label{sect equivalent definitions of duality}

Similarly to \cref{prop tobyextra}, once we establish the characterisation of the measurement duality map, it becomes clear that other physical properties are necessarily related in the dual systems, in particular the partition functions and entropies.
Certain dualities, such as Bosonisation and Kramer-Wannier, are imposed on the level of partition functions.
Therefore, while the measurement duality maps are candidates to describe these types of dualities, this one-way implication does not preclude other mathematical mappings that preserve thermal properties and describe these dual phenomena.

This section establishes the reverse equivalence: duality definitions based on the preservation of partition functions or entropies are in fact essentially equivalent to the measurement duality maps defined in the previous section.
This connection is particularly interesting to unify different dualities on the level of partition functions, measurement outcomes and entropies.

The connection between partition functions and the spectra arises from a transformation of a partition function equality into an infinite sequence of polynomials in the charges (e.g. $\beta$ for the Hamiltonian).
This sequence is shown to converge in the limit to a relation between the $\ell_\infty$ norms of the spectra.
A recursive application of this argument then implies the preservation of the spectral sets themselves.
The connection between entropy preserving and spectrum preserving is perhaps more surprising, and leads to a novel result concerning the characterisation of entropy preserving maps up to an additive constant.

\subsection{Partition function duality}

Examples of physical dualities suggest that it is common for a duality to be defined in terms of partition functions (or equivalently free energy), rather than observables, particularly when considering classical thermodynamics.
This motivates considering a different definition of duality, formulated in terms of preserving partition functions rather than measurement outcomes:

\begin{restatable}[Thermal duality map]{defn}{axioms2}
\label{Thermal duality axioms}
A \emph{thermal duality map}, $\Phi_t: \textup{Herm}_n \mapsto\textup{Herm}_{m}$ satisfies
\begin{enumerate}[(i)]
\item $\forall$ $a_i\in \textup{Herm}_n$, $p_i\in [0,1]$ with $\sum_i p_i = 1:$\\
 \[\Phi_t\left(\sum_i p_i a_i \right) = G\left(\sum_i p_i a_i \right)\sum_i g( a_i)h(p_i) \Phi_t(a_i);\]
\item $\forall$ $A\in \textup{Herm}_n$ and all $J_A>0$, $J_A\in\mathbb{R}$:\\
  \[\alpha\trace\left[e^{-J_Af(A)A} \right] = \trace \left[ e^{-J_A\Phi_t(A)}\right]\]
   for some constant $\alpha>0$.
\end{enumerate}
The scaling functions $f$, $G$, $g$: $\textup{Herm}_n\mapsto \mathbb{R}$, are Lipschitz on any compact subset of $\textup{Herm}_n$ and map to zero iff the input is the zero operator.
Where as $h$: $[0,1] \mapsto [0,1]$ where $\sum_i h(p_i) = 1$ iff $\sum_i p_i = 1$.
\end{restatable}

The convexity condition is the same as in \cref{Measurement duality axioms}, as is the motivation.
The second axiom captures how the thermal physics of the two systems are related.
The simplest physical example of this is the Hamiltonian of the system, $H$, with inverse temperature, $\beta$, acting as the corresponding charge $J_H$.
However, if the duality is to be complete, this relationship should also hold for other source terms in the partition function $\trace\left[-\beta H + \sum_i J_{A_i}A_i \right]$ to relate both the thermal properties and correlations of the two systems.
We must again allow the freedom of rescaling the values of the charges in the dual system by an operator-dependent scaling function $f$, since this is something that could be done operationally.
Equating these generalised partition functions for all values of the charges is mathematically equivalent to~(ii), since trivially all but one selected charge can be set to $0$ in tern.

The following result demonstrates that a map preserving partition functions up to a physical rescaling as in \cref{Thermal duality axioms} necessarily preserved the spectra up to the same rescaling.

\begin{restatable}[Maps preserving partition functions preserve spectra]{thm}{lm partition1}\label{lm partition1}
Given a map $\Phi_t: \textup{Herm}_n \mapsto\textup{Herm}_{m}$ such that $\forall$ $A\in \textup{Herm}_n$ and all $J_A>0$, $J_A\in\mathbb{R}$
\[\alpha\trace\left[e^{-J_Af(A)A} \right] = \trace \left[ e^{-J_A\Phi_t(A)}\right] \text{ for some constant } \alpha>0,\]
where $f:$ $\textup{Herm}_n\mapsto \mathbb{R}$ is a scaling factor.
$\forall$ $A\in \textup{Herm}_n$, $\Phi_t$ satisfies,
\[\textup{spec}\left[\Phi_t(A)\right] = f(A) \textup{spec}[A].\]
\end{restatable}
\begin{proof}
Initially let $\text{spec}\left[ A \right] = \{\lambda_i \}$ and $\text{spec}\left[ \Phi_t(A)\right] = \{\mu_i \}$ and relate their "partition functions" as in the theorem statement
\begin{align}\label{eqn hold for all A}
\alpha\sum_i e^{-J f(A)\lambda_i}  = \alpha\trace \left[e^{-J f(A)A} \right] =&  \trace \left[e^{-J \Phi_t(A)} \right] =\sum_j e^{-J \mu_j}.
\end{align}
Expanding the exponential using the Maclaurin series, $e^x = \sum_{k=0}^\infty \frac{x^k}{k!}$, which converges for all $x$, gives
\begin{equation} \label{analy_funciton}
\alpha\sum_i^{\dim[A] } \sum_{k=0}^\infty \frac{(-J f(A)\lambda_i)^k}{k!} = \sum_j^{\dim\left[\Phi_t(A)\right] } \sum_{k=0}^\infty \frac{(-J \mu_j)^k}{k!} .
\end{equation}
For the above polynomials to be equal at all values of the charge $J$, the coefficients for each power of $J$ must be equal\footnote{Since \cref{analy_funciton} is an analytic function, see e.g.\ \cite{Atkinson} p133.}.
Equating the $J^0$ coefficients fixes the relationship between the dimensions:
\begin{equation}\label{eqn relation between dims}
\alpha\dim [A ]= \dim\left[\Phi_t(A)\right].
\end{equation}
Therefore the operators $A$ and $\Phi_t(A)$ may act on Hilbert spaces of different dimension (i.e. $n\neq m$).
However, \cref{eqn relation between dims} implies $\alpha$ is a positive rational so we set $\frac{x}{y}:=\alpha$ with $x,y\in\mathbb{Z}^+$ coprime in the following.

For a given $A$, the remaining equalities generate an infinite system of polynomials in $\{\mu_i \}_{i=1}^{\dim\left[\Phi_t(A)\right]}$,
\begin{gather}
  \forall p \in \mathbb{Z}^+ \notag \;:\;
  \frac{x}{y} \sum_{i=1}^{ \dim [A ]} \left(f(A)\lambda_i\right)^p = \sum_{i=1}^{\dim\left[\Phi_t(A)\right]} \mu_i^p. \qquad \label{sim eqns for prop Z}
\end{gather}
Manipulating the sum to remove the multiplicative factors we have $\forall p \in \mathbb{Z}^+$,
\begin{equation}\label{eqn vect of same size}
 \sum_{i=1}^{ x\dim [A ]} \left(f(A)\lambda'_i \right)^p = \sum_{i=1}^{y\dim\left[\Phi_t(A)\right]} \mu'^p_i,
\end{equation}
where we define new vectors $\lambda'$, $\mu'$ with elements $\{\lambda'_{(i-1)x+n}\}_{n=1}^{x} = \lambda_i$ and $\{\mu'_{(i-1)y+n}\}_{n=1}^{y} = \mu_i$, indexing the elements of all vectors in non-decreasing order.

The summations in \cref{eqn vect of same size} now each contain the same number of terms and thus, for even $p=2\varrho$, we can interpret the above as equating the $p$-norms of two $(x\dim [A ] = y\dim\left[\Phi_t(A)\right])$-dimensional vectors:
\begin{equation}\label{eqn p-norms}
  \left(\sum_{i=1}^{ x\dim[A]} \abs{f(A)\lambda'_i}^{2\varrho}\right)^{1/2\varrho}
  = \left(\sum_{i=1}^{y\dim[\Phi_t(A)]} \abs{\mu'_i}^{2\varrho}\right)^{1/2\varrho}.
\end{equation}
Taking the limit $\varrho\rightarrow \infty$, this converges to the $\ell_\infty$ norm of both sides, i.e.\ we can equate the elements of maximum absolute value in each vector: $\max_i \abs{f(A)\lambda_i' }= \max_i \abs{\mu_i'}$.

Now, subtracting $(\max_i f(A)\lambda_i')^{2\varrho} = (\max_i\mu_i')^{2\varrho}$ from both sides of \cref{eqn p-norms}, we obtain an analogous set of $p$-norm equalities but for vectors with length reduced by 1, with the maximum elements removed.
Applying this argument recursively, we conclude that the vectors $f(A)\lambda'$ and $\mu'$ must have identical components up to signs.

The linear variant of \cref{eqn vect of same size} rules out the case where the components $\lambda'$ and $\mu'$ have different signs:
\begin{align}\label{eqn removing minuses}
 \sum_{i=1}^{ x\dim [A ]} f(A)\lambda'_i  = \sum_{i=1}^{y\dim\left[\Phi_t(A)\right]} \mu'_i = \sum_{i=1}^{x\dim [A ]} \pm f(A)\lambda'_i.
\end{align}
This follows as \cref{eqn hold for all A} must hold for all Hermitians $A$, including those with with only positive eigenvalues.
Any term in the sum being negated on the right hand side of \cref{eqn removing minuses} would produce a strictly smaller total than that of the left hand side, therefore,
\begin{equation}\label{eqn primed vectors equal}
\mu' = f(A)\lambda'.
\end{equation}

It remains to use $\lambda'$ and $\mu'$ to find the relation between the original eigenvalue vectors $\lambda$ and $\mu$ (potentially of different lengths).
Choose an $A$ with non-degenerate spectrum, and consider the two smallest eigenvalues of $A$.
We have
\begin{align}
&\lambda_1 = \lambda'_{x} = \frac{\mu'_{x}}{f(A)}\\
&\lambda_2= \lambda'_{x+1} = \frac{\mu'_{x+1}}{f(A)}.
\end{align}
Since $A$ has non-degenerate spectrum, we have $\mu'_{x}\neq \mu'_{x+1}$.
But $\{\mu'_{(i-1)y+n}\}_{n=1}^{y}$ are equal for all $i$ by definition of $\mu'$.
Thus $x\geq y$ and $y=1$, since $x$ and $y$ are coprime.
Hence $\dim \Phi_t(A)$ must be at least as large as $\dim A$ and $\alpha\in\mathbb{Z}^+$.

\cref{eqn primed vectors equal} and $\alpha\in\mathbb{Z}^+$ implies the set equality $\{\mu_i\}_{i=1}^{\alpha\dim[A]}=f(A)\{\lambda_i\}_{i=1}^{\dim[A]}$, where each element of $\mu$ is alpha-fold degenerate.
The two spectra are thus proportional and the proof is complete.
\end{proof}

A simple corollary of this result is that the thermal duality map defined above is equivalent to the previously studied and characterised measurement duality map:

\begin{restatable}{col}{col thermal}\label{col thermal}
The set of maps that describe a thermal duality is equal to the set of maps describing a measurement duality, such that \cref{Measurement duality axioms} and \cref{Thermal duality axioms} are equivalent.
Therefore, thermal duality maps are also of the form,
\[ \Phi_t(A) = f(A) U \left(A^{\oplus p} \oplus \overline{A}^{\oplus q} \right) U^\dagger,\]
where $p,q$ are non-negative integers, $U$ is a unitary transformation and $\bar{A}$ represents the complex conjugate of $A$.
\end{restatable}
\begin{proof}
Recall that a measurement duality map is defined by two conditions
\begin{enumerate}[(I.)]
\item $\Phi_s \left(\sum_i p_i a_i \right) = G\left(\sum_i p_i a_i \right) \sum_i g(a_i)h(p_i)\Phi_s(a_i)$;
\item $\mathrm{spec}[\Phi_s(A)] = f(A) \mathrm{spec}[A]$.
\end{enumerate}
and a partition function duality is also defined by two conditions,
\begin{enumerate}[(i.)]
\item $\Phi_t \left(\sum_i p_i a_i \right) = f\left(\sum_i p_i a_i \right) \sum_i \frac{p_i}{f(a_i)}\Phi_t(a_i)$;
\item $\alpha\trace\left[e^{-J_Af(A)A} \right] = \trace \left[ e^{-J_A\Phi_t(A)}\right] \text{ for some constant } \alpha>0$.
\end{enumerate}
(I) and (i) are identical statements.
From \cref{lm partition1} (ii) implies (II) where the degeneracy of the spectrum is given by $\alpha$ which was shown to necessarily be $\alpha \in \mathbb{Z}^+$.
The reverse implication (II) implies (ii) can be shown.
Given the measurement duality characterisation, the spectrum is not only equal but each eigenvalue has degeneracy $p+q = m/n \in \mathbb{Z}^+$.
Therefore,
\begin{equation}
\trace \left[ e^{-J_A\Phi_s(A)}\right]  = \frac{m}{n}\trace \left[ e^{-J_Af(A)A}\right],
\end{equation}
where we can equate $\alpha = m/n$.

All of the conditions have been shown to be equivalent therefore the two definitions of duality describe the same set of maps.
\end{proof}

\subsection{Entropic duality}

A third and final viewpoint is to consider entropic dualities.

\begin{restatable}[Entropic duality map]{defn}{axioms3}
\label{Entropic duality axioms}
An \emph{entropic duality map}, $\Phi_e: \textup{Herm}_n \mapsto\textup{Herm}_{\alpha n}$ and $\Phi_e: \mathcal{S}(\mathcal{H}_n)\mapsto \mathcal{S}(\mathcal{H}_{\alpha n})$ satisfies
\begin{enumerate}[(i)]
\item  $\forall$ $a_i\in \textup{Herm}_n$, $p_i\in [0,1]$ with $\sum_i p_i = 1:$
  \[\Phi_e\left(\sum_i p_i a_i \right) = \sum_i p_i \Phi_e(a_i);\]
\item $\forall\rho \in \mathcal{S}(\mathcal{H}):$
  \[S(\Phi_e(\rho)) = S(\rho) + \log \alpha;\]
\item $\Phi_e(0)=0$.
\end{enumerate}
\end{restatable}

The justification for the convexity condition is unchanged.
However, the map is additionally constrained to map states to states (positive operators with unit trace) to meaningfully examine the behaviour of dual entropies.
An immediate consequence of this is a simplification of the previously allowed generalised convexity to standard convexity.

The second axiom captures how the entropies of corresponding states are related.
In trivial examples of dual states in different sized spaces, there is additional entropy arising from the additional degrees of freedom in the larger state space.
This gives an additive offset that depends on the Hilbert space dimension in the entropy relation.
For example, if states $\rho$ are mapped to the (trivially) dual states $\Phi(\rho) = \rho\otimes 1/d$, the entropy of the dual state picks up an additional additive contribution: $S(\Phi(\rho)) = S(\rho) + d$.

More generally, for a $d_1$-dimensional maximally mixed state to be dual to the maximally mixed state in $d_2>d_1$ dimensions, the required entropy relation is
\begin{align}
S\left(\frac{1}{d_2}\mathbb{I}_{(d_2\times d_2)} \right) & = \log d_2\\
& = \log \frac{d_2 d_1}{d_1}\\
& = S\left(\frac{1}{d_1}\mathbb{I}_{(d_1\times d_1)} \right) + \log \frac{d_2}{d_1}.
\end{align}
Then $\alpha = d_2/d_1$ and we can identify $\log \alpha$ as a constant entropy offset arising from the different Hilbert space dimensions.

Entropies in quantum information theory express the information content or entanglement of systems.
For example, in holographic dualities such as AdS/CFT there are relationships between the entropy of corresponding states (the Ryu-Takayanagi formula \cite{Ryu2006}).
However, the above definition concerns the global entropy of states and not entanglement entropy of reduced states.
Therefore the \emph{state dependent} additive entropy that appears in the Ryu-Takayanagi formula does not contradict the \emph{state independent} additive entropy we assert, since the latter does not refer to the entropy of a reduced state but rather a state on the full Hilbert space.

Similarly to the previous section we arrive at a characterisation of entropic duality maps by demonstrating that a map that preserves entropies is necessarily spectrum preserving.
To show this result we first need some technical lemmas.

\begin{restatable}[Entropy of mixtures of mixed states] {lemma}{lmentropyofmixturesofmixed}
\label{lm entropy of mixtures of mixed}
Given a density operator, $\rho_\mathcal{A} = \sum_{x=1}^k p_x \rho_x$, that is a probabilistic mixture of mixed states $\rho_x$, with $p_x\in[0,1]$ and $\sum_x p_x =1$.
The von Neumann entropy of $\rho_\mathcal{A}$ obeys the following equality,
\[S(\rho_\mathcal{A}) = \sum_x p_x S(\rho_x) -\sum_x p_x \log p_x,\]
if and only if $\rho_x$ have orthogonal support.
I.e. $\trace [\rho_x \rho_y]=0$ for all $x\neq y$.
\end{restatable}

\begin{restatable}[Pure states mapped to orthogonal density matrices]{lemma}{lmorthogonaldensitymatrices}
\label{lm orthogonal density matrices}
Let $\{ \sigma_i \}_{i=1}^d$ be a set of orthogonal pure states that forms a basis in $\mathcal{H}$, with $\sigma_i \in P_1(\mathcal{H})$.
Let the map $\phi: \mathcal{S}(\mathcal{H}_n)\mapsto \mathcal{S}(\mathcal{H}_{\alpha n})$, be
\begin{enumerate}[(a)]
\item \label{additive entropy relation} entropy preserving up to an additive constant, $S(\phi(\rho)) = S(\rho) + \log \alpha$;
\item convex, $\phi(t\rho + (1-t)\sigma) = t \phi(\rho) + (1-t) \phi(\sigma)$. Where $t\in[0,1]$ and $\rho,\sigma \in \mathcal{S}(\mathcal{H})$.
\end{enumerate}
The image of this set under the map is a new set, $\{\phi(\sigma_i) \}_{i=1}^d$, with orthogonal support.
\end{restatable}

For the proof of \cref{lm entropy of mixtures of mixed} (\cref{lm orthogonal density matrices}) see \cref{appen proof 1} (\cref{appen proof 2}) respectively.
Now that an orthogonal basis in the dual system is established, we can show the entropy preserving map is necessarily spectrum preserving.
Not that in this case the map is only transforming between states, how the map acts on the full Hermitians is the subject of the next result.

\begin{restatable}[Entropy preserving implies spectrum preserving on positive normalised Hermitian operators]{thm}{lmentropyspecpres}
\label{lm entropy spec pres}
A map $\phi: \mathcal{S}(\mathcal{H}_n)\mapsto \mathcal{S}(\mathcal{H}_{\alpha n})$, that is
\begin{enumerate}[(a)]
\item entropy preserving up to an additive constant: $S(\phi(\rho)) = S(\rho) + \log \alpha$;
\item convex: $\phi(t\rho + (1-t)\sigma) = t \phi(\rho) + (1-t) \phi(\sigma)$. Where $t\in[0,1]$ and $\rho,\sigma \in \mathcal{S}(\mathcal{H})$
\end{enumerate}
will transform the spectrum of the density operator in the following way
\begin{align*}
&\textup{spec}[\rho] = \{\lambda_1, ...,\lambda_d \}\\
&\textup{spec}\left[\phi(\rho)\right] = \left\{\frac{\lambda_1}{\alpha},...,\frac{\lambda_d}{\alpha} \right\} \label{eqn spec of mapped density op}
\end{align*}
where every eigenvalue in the spectrum of $\phi(\rho)$ has multiplicity $\alpha$.
\end{restatable}

\begin{proof}
The first step in the proof is to show that the image of the pure states $\{\phi(\sigma_i)\}_{i=1}^d$ -- which by \cref{lm orthogonal density matrices} is known to have orthogonal support -- has $\alpha$ non-zero eigenvalues all equal to $1/\alpha$.
Using the entropy preserving property of the map: $S(\phi(\sigma_i))=\log \alpha$.
Since $\log \alpha$ is the maximal entropy of a Hilbert space of dimension $\alpha$, it follows that $\phi(\sigma_i)$ must have at least $\alpha$ non-zero eigenvalues, i.e. $\text{Rank}\left[ \phi(\sigma_i)\right]\geq \alpha$ for all $i$.

As a consequence of orthogonality, the rank summation of $d$ mixed states, $\phi(\sigma_i)$, will be upper bounded by the dimension of the Hilbert space the density matrices act in:
\begin{equation}
\sum_{i=1}^d \text{Rank} \left[\phi(\sigma_i) \right] \leq \alpha d.
\end{equation}
It follows that $\text{Rank}\left[\phi(\sigma_i)\right] = \alpha$ for all $i$.
Together with the entropy $S(\phi(\sigma_i))=\log \alpha$ it follows that the non-zero eigenvalues must be flat and $\text{spec}\left[\phi(\sigma_i) \right] = \{1/\alpha,0 \}$. \newline

It is then simple to extend to the full result.
Any state in $\rho \in \mathcal{S}(\mathcal{H})$ can be written as a linear combination of pure states $\rho = \sum_{i=1}^d \lambda_i \sigma_i$ where due to normalisation  $\sum_{i=1}^d \lambda_i = 1$.
Using the convexity property of the map
\begin{equation}
\phi\left(\sum_{i=1}^d \lambda_i \sigma_i \right) = \sum_{i=1}^d \lambda_i \phi(\sigma_i).
\end{equation}
From \cref{lm entropy of mixtures of mixed} $\{ \phi(\sigma_i)\}$ have orthogonal support and therefore $\text{spec}\left[\phi(\sigma_i) \right] = \{1/\alpha,0 \}$.
Therefore the spectrum of $\phi(\rho)$ will be $\{\lambda_1/\alpha, \lambda_2/\alpha, ... , \lambda_d/\alpha\}$ each with multiplicity $\alpha$.
\end{proof}

Armed with a link between entropy preserving and spectral preserving on positive Hermitians with unit trace, we can now look to characterising the entropic dual maps on the full Hermitian space.
We show that the entropic definition of duality is only slightly less general than the others.
This originates from the normalisation of elements of $\mathcal{S}(\mathcal{H})$ whereby since the operator map the is restricted to map states to states the rescaling is limited.
The following result characterised entropic duality maps and describes the almost equivalence to the two other types of duality map we have studied.

\begin{restatable}{thm}{colaxdualentr}
\label{col ax_dual_entr}
Every entropic duality map $\Phi_e$ is a measurement/thermal duality map where $f(A)=1/\alpha$ for all $A\in \textup{Herm}_n$ and therefore has the form
\[\Phi_e(A) = \frac{1}{\alpha}U \left(A^{\oplus p} \oplus \overline{A}^{\oplus q} \right) U^\dagger,\]
for some unitary $U$ and $p,q\in\mathbb{Z}^+$.
Conversely if $\Phi$ is a measurement/thermal duality map then the related map
\[
\Phi'_e(A) := \begin{cases}
\frac{\Phi(A)}{\alpha f(A)} & \textup{for } A\in\textup{Herm}_n \neq 0\\
\Phi(A) & \textup{for } A=0,
\end{cases}
\]
is an entropic duality map.
\end{restatable}

\begin{proof}
\cref{col thermal} states that measurement and thermal duality maps are equivalent.
Therefore this proof can focus on demonstrating a relationship between $\Phi_e$ and measurement duality maps and the connection to thermal duality maps in identical.

Recall that an entropic duality map is defined by three conditions
\begin{enumerate}[(i.)]
\item $\Phi_e\left(\sum_i p_i a_i \right) = \sum_i p_i \Phi_e(a_i)$;
\item $S(\Phi_e(\rho)) = S(\rho) + \log \alpha$;
\item $\Phi_e(0)=0$;
\end{enumerate}
and a measurement duality map is defined by two conditions
\begin{enumerate}[(I.)]
\item $\Phi_s \left(\sum_i p_i a_i \right) = f\left(\sum_i p_i a_i \right) \sum_i \frac{p_i}{f(a_i)}\Phi_s(a_i)$;
\item $\mathrm{spec}[\Phi_s(A)] = f(A) \mathrm{spec}[A]$.
\end{enumerate}
We have used \cref{lm Constrained scale functions} to replace the original weakened convexity condition with the constrained convexity condition that equivalently defines the map.
Additionally, restricting to the case where $f(A) = \frac{1}{\alpha}$ then condition I becomes,
\begin{equation}
\Phi_s\left( \sum_i p_i a_i\right) = \frac{1}{\alpha} \sum \alpha p_i \Phi_s(a_i) = \sum_i p_i \Phi_s(a_i)
\end{equation}
such that it is manifestly equivalent to (i) for this choice of scale function.

All that is left to do for the first statement is to show that a map obeying (i)-(iii) is spectrum preserving for all Hermitians.
The first step is to show that (i) \& (iii) implies the map, $\Phi_e$, is real linear.
This follows from the same argument laid out in the proof of \cite{Cubitt2019} Theorem 4.
For any real negative $\lambda$ set $p=\frac{\lambda}{\lambda-1}>0$, $A\in\text{Herm}_n$ and $B = \frac{pA}{(p-1)}= \lambda A$.
Using (i) and (iii) together:
\begin{align}
\Phi_e(pA + (1-p)B) & = \Phi_e(0) = 0\\
& = p \Phi_e(A) + (1-p)\Phi_e(\lambda A).
\end{align}
Therefore $\lambda \Phi_e(A) = \Phi_e(\lambda A)$.
Repeating this logic for $\lambda A$ gives $\lambda^2 \Phi_e(A) = \Phi_e(\lambda^2 A)$ and hence homogeneity for all real scalars.
Then combining (i) with homogeneity gives real linearity of $\Phi_e$, i.e.
\begin{align}
\Phi_e\left(\sum_i p_i \lambda a_i\right) & = \sum_i p_i \Phi_e(\lambda a_i) = \sum_i \lambda p_i a_i, \label{eqn real linear entropy state map}
\end{align}
for $(\lambda p_i)\in \mathbb{R}$ and $a_i\in \text{Herm}_n$.

The entropic duality map restricted to $\mathcal{S}(\mathcal{H}_n)$ satisfies the conditions of \cref{lm entropy spec pres} and therefore $\Phi_e$ preserves the spectra of positive Hermitians with unit trace (up to a renormalisation).
The transformation of the spectra of $M\not \in\mathcal{S}(\mathcal{H})$ by $\Phi_e$ is shown by building up from $\sigma,\rho \in \mathcal{S}(\mathcal{H})$ using $\Phi_e(a\rho+b\sigma)= a \Phi_e(\rho) + b \Phi_e(\sigma)$.
First note that any Hermitian operator can be written in a spectral decomposition $M = \sum_i \nu_i \ket{\psi_i}\bra{\psi_i}$.
Splitting the decomposition up into two sums over the positive and negative eigenvalues respectively,
\begin{align}
M & = \sum_{\nu_i>0}\nu_i \ket{\psi_i}\bra{\psi_i} + \sum_{\nu_i<0}\nu_i \ket{\psi_i}\bra{\psi_i}\\
& = M_+ + M_-\\
& = c_+ \rho_+ + c_- \rho_-,
\end{align}
where $\rho_{+/-} = \frac{M_{+/-}}{\trace(M_{+/-})}$ and $c_{+/-}=\trace(M_{+/-})$.
Therefore
\begin{equation}
\Phi_e(M) = c_+ \Phi_e(\rho_+) + c_- \Phi_e(\rho_-).
\end{equation}
Since $\rho_+$ and $\rho_-$ are orthogonal it follows from \cref{lm entropy spec pres,lm orthogonal density matrices} that the spectrum of $M$, $\{ \nu_i\}_{i=1}^d$ transforms as
\begin{equation}
\textup{spec}\left[\Phi_e(M)\right] =\frac{1}{\alpha} \left\{\nu_1,...,\nu_d\right\},
\end{equation}
where every eigenvalue in the new spectrum has multiplicity $\alpha$.\\

The converse statement is simple to demonstrate.
For all $A\in \textup{Herm}_n$, $\Phi(A) = \Phi(A)^\dagger$ and since $f(A)\in\mathbb{R}$ it follows that $\Phi'_e$ also preserves Hermiticity.
Using the simplified convexity axiom from \cref{lm Constrained scale functions} for $\Phi$, and substituting for $\Phi'_e$, it is easy to see that this map is convex as in (i) of the definition of entropic duality maps.
Finally using spectrum preservation of $\Phi$,
\begin{equation}
\text{spec}\left[\Phi'_e(\rho) \right] = \frac{1}{\alpha} \text{spec}\left[ \rho\right],
\end{equation}
for a state $\rho\in\text{Herm}_n$, where each eigenvalue has $\alpha$ copies.
$S(\rho) = \sum_i \eta_i \log \eta_i$ where $\{ \eta_i\}$ are the eigenvalues of $\rho$.
Therefore the entropy of the mapped state is,
\begin{align}
S(\Phi'_e(\rho)) &= -\alpha \sum_i \left(\frac{\eta_i}{\alpha} \right) \log \left(\frac{\eta_i}{\alpha} \right)\\
& = - \sum_i \eta_i \log \eta_i + \sum_i \eta_i \log \alpha\\
& = S(\rho) + \log \alpha,
\end{align}
and the second axiom of \cref{Entropic duality axioms} is satisfied by the map.
The third axiom follows immediately from $\Phi'_e(0) := \Phi(0)=0$, giving the converse statement.
\end{proof}

\subsubsection{Extension to Wigner's theorem: a new characterisation of entropy preserving maps}\label{sect Wigners}

The above connection between entropy preserving and spectrum preserving axioms is notable since there is indepdendant interest in characterising entropy preserving maps.
While it is well-known that a unitary or anitunitary\footnote{An antiunitary operator is a bijective antilinear map $W:\mathcal{H}\mapsto\mathcal{H}$ of a complex Hilbert space such that $\langle Wx, Wy \rangle = \overline{\langle x, y \rangle}$ for all $x,y\in\mathcal{H}$ where the overline denotes complex conjugation.} transformation leaves the entropy invariant, the reverse implication is false without additional information.
Previous work, that traces its origins back to Wigner's celebrated theorem~\cite{Wigner}, has shown that by demanding additional constraints on entropy preserving maps, the maps are entirely characterised by either a unitary or antiunitary transformation.
\begin{restatable}[Previous entropic map characterisations]{prop}{propotherentropy}\label{prop prev ent}
Given a surjective map on states $\phi: \mathcal{S}(\mathcal{H})\mapsto \mathcal{S}(\mathcal{H})$ where the Hilbert space $\mathcal{H}$ has dimension $n$,
\begin{enumerate}
\item \cite{He2012} For all $\rho \in \mathcal{S}(\mathcal{H})$, $\forall \lambda \in [0,1]$
\[S(\lambda \rho + (1-\lambda)\I/n) = S(\lambda \phi(\rho) + (1-\lambda)\I/n)\]
iff $\phi(\rho) = W \rho W^*$ for some unitary or anti-unitary operator $W$.
\item \cite{He2015} For all $\rho,\sigma \in \mathcal{S}(\mathcal{H})$, $\forall \lambda \in [0,1]$
\[S(\rho + (1-\lambda)\sigma) = S(\lambda \phi(\rho) + (1-\lambda)\phi(\sigma))\]
iff $\phi(\rho) = W \rho W^*$ for some unitary or anti-unitary operator $W$.
\item \cite{Molnar2008} For all $\rho,\sigma \in \mathcal{S}(\mathcal{H})$,
\[S(\rho||\sigma) = S(\phi(\rho)||\phi(\sigma))\]
iff $\phi(\rho) = W \rho W^*$ for some unitary or anti-unitary operator $W$.
\end{enumerate}
\end{restatable}
These can be translated into the language used in our characterisation theorem by noting that for any antiunitary operator $W$, the operator $WK$, where $K$ is the complex conjugation operator, is unitary.
Therefore in \cref{prop prev ent} either $W$ is unitary which corresponds to $p=1$, $q=0$ or if $W$ is anti-unitary, for some unitary $U$, $\phi(\rho) = U \overline{\rho} U^\dagger$ corresponding to $p=0$, $q=1$.
Hence all maps in the above proposition are found to be encodings with $p+q\leq 1$.

However, to our knowledge maps preserving entropy up to an additive constant have not been studied in the literature.
A direct consequence of \cref{col ax_dual_entr} is a natural extension of these previous generalisations of Wigner's theorem arises.
A map $\Phi: \mathcal{S}(\mathcal{H})\mapsto \mathcal{S}(\mathcal{H}^{\oplus \alpha})$ is convex,
\begin{equation}
\Phi(\sum_i p_i \rho_i) = \sum_i p_i \Phi(\rho_i),
\end{equation}
and entropy preserving up to an additive constant
\begin{equation}
S(\Phi(\rho_i)) = S(\rho_i) + \log \alpha,
\end{equation}
for all $\rho_i \in \mathcal{S}(\mathcal{H})$, $p_i \in [0,1]$ with $\sum_i p_i = 1$, where $\alpha\in\mathbb{Z}_{\geq 0}$;
iff $\Phi$ is of the form,
\begin{equation}
\Phi(\rho) = U \left(\bigoplus_{i=1}^p V_i \rho V_i^\dagger \oplus \bigoplus_{i=p+1}^{p+q}W_i\rho_i W_i^\dagger \right)U^\dagger
\end{equation}
for some unitaries $U,V_i$ and antiunitaries $W_i$ acting on $\mathcal{H}$, where $p,q\in\mathbb{Z}_{\geq 0}$ and $p+q = \alpha$.

Whereas previous characterisations of entropy preserving maps reduce to Wigner's theorem, by taking a different route via Jordan and $C^*$~algebra techniques we show that the entropic additive constant is precisely the additional freedom that allows the maps to admit a direct sum of both unitary and antiunitary parts.

\section{Approximate dualities}\label{sect Approximate dualities}

So far only exact dualities have been considered.
However, more general definitions of duality are needed in order for this framework to be practical.
This section defines how to extend the ideas of exact duality maps to allow for approximations and restrictions to a subspace.
Here approximate refers to the physics of the two systems being equal up to some error, the approximate equivalence holds within the full subspace.
However, the simulation within a subspace corresponds to the other type of `approximate' duality discussed in the Bosonisation example, where the equivalence only holds in some regime e.g. the low energy regime.

\begin{restatable}[$(\mathcal{S}, \epsilon)$-Duality]{defn}{defnduality}
\label{def approx Duality}
$\tilde{\Phi}$: $\textup{Herm}_n\mapsto \textup{Herm}_m$ is a $(\mathcal{S}, \epsilon)$-approximate duality map if $\exists$ a duality map $\Phi$ such that $\forall A\in\textup{Herm}_n$, the action of $\tilde{\Phi}$ restricted to the subspace $\mathcal{S}$ is close to the action of $\Phi$:
\[\norm{\left.\tilde{\Phi}(A)\right|_{\mathcal{S}} - \Phi(A)} \leq k(A) \epsilon ,\]
for some constant $\epsilon$, where $k:\textup{Herm}_n\mapsto \mathbb{R}_{\geq 0}.$
The duality map is:
\begin{enumerate}[(i.)]
\item \textbf{exact} if $\epsilon=0$;
\item \textbf{unital} if $f(A)=1$ for all $A\in \textup{Herm}_n$.
\end{enumerate}
\end{restatable}

\cite{Cubitt2019} places a large emphasis on local simulations given the focus on Hamiltonian simulation.
Since many-body Hamiltonians of interest are often local, a local encoding will preserve this local structure.
Exact dualities by simple extension are those related to a local encoding,
\begin{restatable}[Local duality map]{defn}{defndualitymap}
\label{defn local duality map}
A local duality map $\Phi: \textup{Herm}_n \mapsto \textup{Herm}_m$ is a duality map i.e. of the form $\Phi(A) = f(A)\mathcal{E}(A)$, where the corresponding encoding $\mathcal{E}$ is a local encoding in the sense of \cite{Cubitt2019} definition 13.
\end{restatable}
\noindent Due to the close relation between duality maps and encodings, we can extend the above definition to focus on approximately local duality maps.
\begin{restatable}[$(\mathcal{S}, \epsilon,\eta)$-Local duality]{defn}{defnlocalduality}
\label{def local approx Duality}
$\tilde{\Phi}$: $\textup{Herm}_n\mapsto \textup{Herm}_m$ is a $(\mathcal{S}, \epsilon, \eta)$-approximately local duality map if it is an $(\mathcal{S}, \epsilon)$-approximate duality map and the exact duality map $\Phi(M) = f(M)V \left(M^{\oplus p}\oplus \overline{M}^{\oplus q} \right) V^\dagger$ in \cref{def approx Duality} is close to a local duality map (\cref{defn local duality map}), $\Phi'(M) = f(M)V' \left(M^{\oplus p}\oplus \overline{M}^{\oplus q} \right) V'^\dagger$, such that $\norm{V - V'}\leq \eta$.
The duality is \textbf{exactly-local} if $\eta=0$.
\end{restatable}

Locality is a natural property to consider, but similar definitions could be equivalently formulated for some other desirable properties, for example particle number conserving.
How these error parameters translate to errors in the physically relevant properties is explored in \cref{sect Physical properties}.

The remainder of this section demonstrates that the definition of duality mappings (and their approximate counterparts), arising from physically motivated axioms, have several desirable properties.
In particular, exact and approximate dualities are shown to compose well.
The choice of extension to approximate mappings is further motivated since the errors defined are shown propagate to physically relevant properties in a controlled way.

\subsection{Similar mappings}\label{sect Similar mappings}

As expected, if two exact duality maps are close the results of applying the maps to the same operator are also close.
Furthermore applying the same mapping to two close operators gives outputs that are close.
This was formalised for encodings in Lemma 19 of \cite{Cubitt2019}, here we show a similar result for duality maps where, unsurprisingly, the "closeness" now also depends on the scaling functions of the maps involved.

First we restate Lemma~18 of \cite{Cubitt2019}, a technical result used in the following proof.

\begin{restatable}{lemma}{lemmacubitt}
\label{lm cubitt1}
Let $A,B: \mathcal{H}\rightarrow \mathcal{H}'$ and $C:\mathcal{H}\rightarrow \mathcal{H}$ be linear maps.
Let $\norm{\cdot}_a$ be the trace or operator norm.
Then,
\begin{equation}
\norm{ACA^\dagger - BCB^\dagger}_a\leq(\norm{A}+\norm{B})\norm{A-B} \hspace{2pt} \norm{C}_a .
\end{equation}
\end{restatable}

\begin{restatable}[Similar exact dualities]{prop}{propsimiliardualities}
\label{prop similar dualities}
Consider two duality maps $\Phi$ and $\Phi'$ defined by
${\Phi(M) = f(M)V\left(M^{\oplus p}\oplus \overline{M}^{\oplus q} \right)V^\dagger}$, ${\Phi'(M)=f'(M)V'\left(M^{\oplus p}\oplus \overline{M}^{\oplus q} \right) V'^\dagger}$, for some isometries $V$, $V'$.
Then for any operators $M$ and $M'$:
\begin{enumerate}[(i)]
\item $\norm{\Phi(M) - \Phi'(M)}\leq\left(\abs{\sqrt{f(M)}} + \abs{\sqrt{f'(M)}} \right) \norm{\sqrt{f(M)}V - \sqrt{f'(M)}V' } \norm{ M}$;
\item $\norm{\Phi(M)-\Phi(M')} = \norm{f(M)M-f(M')M'}$.
\end{enumerate}
\end{restatable}

\begin{proof}
For \textit{(i)} applying \cref{lm cubitt1} gives
\begin{align}
\norm{\Phi(M)-\Phi'(M)}&=\norm{f(M)V\mathbf{M}V^\dagger - f'(M)V'\mathbf{M}V'^\dagger}\\
&\leq \left(\norm{\sqrt{f(A)}V} + \norm{\sqrt{f'(M)}V'} \right) \norm{\sqrt{f(M)}V - \sqrt{f'(M)}V'} \norm{M}\\
&= \left(\abs{\sqrt{f(M)}}+ \abs{\sqrt{f'(M)}} \right) \norm{\sqrt{f(M)}V - \sqrt{f'(M)}V'} \norm{M},
\end{align}
where $\mathbf{M}=M^{\oplus p}\oplus \overline{M}^{\oplus q}$.
The second part is simply
\begin{align}
\norm{\Phi(M) - \Phi(M')}& = \norm{f(M)V\left(M^{\oplus p} \oplus \overline{M}^{\oplus q} \right)V^\dagger - f(M')V\left(M'^{\oplus p} \oplus \overline{M'}^{\oplus q} \right)V^\dagger}\\
&= \norm{f(M)V\left(\left(M-\frac{f(M')}{f(M)}M'\right)^{\oplus p} \oplus \left(\overline{M}-\frac{f(M')}{f(M)}\overline{M'}\right)^{\oplus q} \right)V^\dagger}\\
& = \norm{f(M)M-f(M')M'}.
\end{align}
\end{proof}

\subsection{Composition}\label{sect Composition}

It follows almost directly from \cite{Cubitt2019} Lemma 17 that the composition of two exact duality maps, $\Phi = \Phi_2 \circ \Phi_1$ will itself be an exact duality map, therefore we first restate their result.

\begin{restatable}{lemma}{lemmacubitt2}
\label{lm cubitt2}
If $\mathcal{E}_1$ and $\mathcal{E}_2$ are encodings, then their composition $\mathcal{E}_1\circ\mathcal{E}_2$ is also an encoding,
Furthermore, if $\mathcal{E}_1$ and $\mathcal{E}_2$ are both local, then their composition $\mathcal{E}_1\circ\mathcal{E}_2$ is local.
\end{restatable}

\begin{restatable}[Exact duality map composition]{prop}{propexactdualitycomposition}
\label{prop exact duality composition}
Let $\Phi_1$ and $\Phi_2$ be duality maps.
The composition of these maps, $\Phi = \Phi_2 \circ \Phi_1$ is also a duality map with the valid duality scaling function $f(\cdot)=f_2(\Phi_1(\cdot))f_1(\cdot)$.
Furthermore if the initial dualities were both local, the composition is also local.
\end{restatable}

\begin{proof}
The two duality maps necessarily have the form
\begin{align}
\Phi_1(M) &= f_1(M)V_1 \left(M\otimes P_1 + \overline{M}\otimes Q_1\right) V_1^\dagger ,\\
\Phi_2(M) &= f_2(M)V_2 \left(M\otimes P_2 + \overline{M}\otimes Q_2\right) V_2^\dagger,
\end{align}
where $V_i$ are isometries, $f_i$ are real functions and $P_i$, $Q_i$ are orthogonal projectors.
This leads to a composition of the form,
\begin{equation}
 \! \begin{multlined} (\Phi_2 \circ \Phi_1)(M) = f_2(f_1(M)V_1 \left(M\otimes P_1 + \overline{M}\otimes Q_1\right) V_1^\dagger)f_1(M) \times \\
\hspace{60pt} V_2 \left[ V_1 \left(M\otimes P_1 + \overline{M}\otimes Q_1\right) V_1^\dagger\otimes P_1 \right.\\
+ \left.\overline{V_1 \left(M\otimes P_2 +\overline{M}\otimes Q_2\right) V_1^\dagger}\otimes Q_1\right]V_2^\dagger.
\end{multlined}
\end{equation}
\cref{lm cubitt2} tells us this can be rewritten as,
\begin{align}
(\Phi_2 \circ \Phi_1)(M) = f_2&(f_1(M)V_1 \left(M\otimes P_1 + \overline{M}\otimes Q_1\right) V_1^\dagger) f_1(M)\times \notag\\
& U\left[ M \otimes P + \overline{M}\otimes Q \right] U^\dagger ,
\end{align}
where $U=V_2(V_1 \otimes P_2+ \overline{V_1}\otimes Q_2 + \mathbb{I}\otimes (\mathbb{I}- P_2-Q_2))V_2^\dagger$ is an isometry and $P=P_1\otimes P_2+ \overline{Q}_1\otimes Q_2$, $Q=Q_1\otimes P_2+\overline{P}_1\otimes Q_2$ are new orthogonal projectors.

All that remains is to identify a new scaling function,
\begin{equation}
f(M) = f_2(f_1(M)V_1 \left(M\otimes P_1 + \overline{M}\otimes Q_1\right) W^\dagger)f(M),
\end{equation}
and note that it satisfies the three prerequisites from the definition of a duality map.
The first two are immediate: it maps operators to real scalars and doesn't map to zero unless the operator is zero.
Checking the function is also Lipschitz on compact sets requires slightly more work.

We would like to show for all $B,B'$ in any compact subset there exists a constant $L$ such that,
\begin{equation}
\abs{f_2(\Phi_1(B))f_1(B) - f_2(\Phi_1(B'))f_1(B')} \leq L \norm{B-B'}.
\end{equation}
Breaking this down and using knowledge of $f_1,f_2$,
\begin{align}
&\abs{f_2(\Phi_1(B))f_1(B) - f_2(\Phi_1(B'))f_1(B')} \notag\\
& \leq \abs{f_2(\Phi_1(B))}\abs{f_1(B)-f_1(B')} + \abs{f_1(B')}\abs{f_2(\Phi_1(B))-f_2(\Phi_1(B'))}\\
&\leq \abs{f_2(\Phi_1(B))}L_1 \norm{B-B'}+\abs{f_1(B')}L_2 \norm{\Phi_1(B)-\Phi_1(B')}.\label{eqn using subsets}
\end{align}
Using result (ii) from \cref{prop similar dualities},
\begin{align}
\norm{\Phi_1(B)-\Phi_1(B')} &\leq \abs{f_1(B')}\norm{B-B'} + \abs{f_1(B)-f_1(B')}\norm{B}\\
& \leq \abs{f_1(B')}\norm{B-B'} + L_1\norm{B}\norm{B-B'}.
\end{align}
Therefore,
\begin{align}
&\abs{f_2(\Phi_1(B))f_1(B) - f_2(\Phi_1(B'))f_1(B')} \notag\\
& \leq \left(\abs{f_2(\Phi_1(B))} L_1 + \abs{f_1(B')}L_2 \left(\abs{f_1(B')}+L_1\norm{B} \right)\right)\norm{B-B'}\\
&\leq L \norm{B-B'}.
\end{align}
The function is then a valid rescaling since for all $B,B$ in a compact set there exists a constant $L$ such that,
\begin{equation}
\abs{f_2(\Phi_1(B))} L_1 + \abs{f_1(B')}L_2 \left(\abs{f_1(B')}+L_1\norm{B} \right) \leq L,
\end{equation}
as compactness implies $\norm{B},f_2(\Phi_1(B)),f_1(B)$ can be upper bounded by a constant.

The scale factor is independent of the locality structure so it follows directly from \cref{lm cubitt2} that if the initial dualities were both local the composition is also local.
\end{proof}

This can now be extended to consider how the error parameters translate when two approximately-local duality maps are composed.

\begin{restatable}[Approximate duality composition]{prop}{propapproxdualitycomposition}
\label{prop approx duality composition}
Let $\tilde{\Phi}_1$, $\tilde{\Phi}_2$ be $(\mathcal{S}_i, \epsilon_i, \eta_i)$-approximately local duality maps with corresponding close exact duality maps $\Phi_1(\cdot)=f_1(\cdot)\mathcal{E}_1(\cdot)$, $\Phi_2=f_2(\cdot)\mathcal{E}_2(\cdot)$ respectively.
Their composition $\tilde{\Phi} = \tilde{\Phi}_2\circ \tilde{\Phi}_1$ is a $(\mathcal{S}, \epsilon, \eta)$-approximately local duality map on any compact subset where,
\begin{align}
&\epsilon  = \epsilon_1 + \epsilon_2,\\
&\eta \leq  \eta_1 + \eta_2,\\
&k(A)  = \begin{multlined}k_2\left(\tilde{\Phi}_1\left.(A)\right|_{\mathcal{S}_1}\right) + L_2k_1(A)^2\epsilon_1 \\
+ \Lambda_2 \abs{f_1(A)}\norm{A}k_1(A)+\abs{f_2(\Phi_1(A))}k_1(A).\end{multlined}
\end{align}
Here, $L_2$ is the Lipschitz constant of $f_2$.
Moreover the exact duality that is close to the approximate composition is the composition of exact dualities, $\Phi_2\circ\Phi_1$.
$\mathcal{S}\subseteq \mathcal{S}_2$ is the subspace given by the domain of $\Phi_2$ when the range is restricted to $\mathcal{S}_1$.
\end{restatable}

\begin{proof}
Since $\tilde{\Phi}_1$ and $\tilde{\Phi}_2$ are approximate dualities,
\begin{align}
&\norm{\left.\tilde{\Phi}_1(A)\right|_{\mathcal{S}_1}- \Phi_1(A)} \leq k_1(A)\epsilon_1 \label{eqn first comp dual}\\
&\norm{\left.\tilde{\Phi}_2(A)\right|_{\mathcal{S}_2}- \Phi_2(A)} \leq k_2(A)\epsilon_2.
\end{align}
For $\tilde{\Phi}$ to be an approximate duality it must satisfy an inequality of the following form,
\begin{equation}\label{eqn goal comp approx dual}
\norm{\left.\tilde{\Phi}_2\left(\left.\tilde{\Phi}_1(A)\right|_{\mathcal{S}_1}\right)\right|_{\mathcal{S}_2}- \Phi(A)} \leq k(A)\epsilon,
\end{equation}
for some exact duality $\Phi$, where we have used knowledge of $\mathcal{S}$ to rewrite the restriction.
Exact dualities compose to give a valid exact duality $\Phi_2\circ \Phi_1(A) = f_2(\Phi_1(A))f_1(A)\mathcal{E}_2\circ \mathcal{E}_1(A)$ (see \cref{prop exact duality composition}).
So we take this as $\Phi$ in \cref{eqn goal comp approx dual} and show that the norm difference is bounded by something of the form of the right hand side of \cref{eqn goal comp approx dual}.

Using the knowledge of the composite dualities and the triangle inequality,
\begin{equation}\label{eqn comp 1}
\begin{multlined}
\norm{\tilde{\Phi}_2\left.\left(\tilde{\Phi}_1\left.(A)\right|_{\mathcal{S}_1}\right)\right|_{\mathcal{S}_2} - \Phi_2\circ \Phi_1(A)}  \\
 \leq k_2\left(\tilde{\Phi}_1\left.(A)\right|_{\mathcal{S}_1}\right)\epsilon_2 + \norm{\Phi_2\left(\tilde{\Phi}_1\left.(A)\right|_{\mathcal{S}_1}\right) - \Phi_2\circ \Phi_1(A)}.
\end{multlined}
\end{equation}

The second term in \cref{eqn comp 1} can be broken down using the similar exact dualities result (ii) from \cref{prop similar dualities},
\begin{align}
& \norm{\Phi_2\left(\tilde{\Phi}_1\left.(A)\right|_{\mathcal{S}_1}\right) - \Phi_2\circ \Phi_1(A)} \notag\\
&= \norm{f_2(\tilde{\Phi}_1\left.(A)\right|_{\mathcal{S}_1})\tilde{\Phi}_1\left.(A)\right|_{\mathcal{S}_1} - f_2(\Phi_1(A))\Phi_1(A)}\\
& \begin{multlined}\leq \Bigl|\Bigl|\left(f_2(\tilde{\Phi}_1\left.(A)\right|_{\mathcal{S}_1})-f_2(\Phi_1(A))+f_2(\Phi_1(A))\right)\left(\tilde{\Phi}_1\left.(A)\right|_{\mathcal{S}_1}- \Phi_1(A)\right.\\
 \left.+\Phi_1(A)\right) - f_2(\Phi_1(A))\tilde{\Phi}_1\left.(A)\right|_{\mathcal{S}_1}\Bigl|\Bigl|\end{multlined}\\
& \begin{multlined}\leq \abs{f_2(\tilde{\Phi}_1\left.(A)\right|_{\mathcal{S}_1})-f_2(\Phi_1(A))}\left(\norm{\tilde{\Phi}_1\left.(A)\right|_{\mathcal{S}_1}- \Phi_1(A)}+\norm{\Phi_1(A)}\right) +\\
 \abs{f_2(\Phi_1(A))}\norm{\tilde{\Phi}_1\left.(A)\right|_{\mathcal{S}_1}-\Phi_1(A)}\end{multlined}\\
& \leq \abs{f_2(\tilde{\Phi}_1\left.(A)\right|_{\mathcal{S}_1})-f_2(\Phi_1(A))}\left(k_1(A)\epsilon_1+\abs{f_1(A)}\norm{A}\right) + \abs{f_2(\Phi_1(A))}k_1(A)\epsilon_1.
\end{align}
Substituting this back gives,
\begin{equation}
\begin{multlined}
\norm{\tilde{\Phi}_2\left.\left(\tilde{\Phi}_1\left.(A)\right|_{\mathcal{S}_1}\right)\right|_{\mathcal{S}_2} - \Phi_2\circ \Phi_1(A)}
 \\
 \leq k_2\left(\tilde{\Phi}_1\left.(A)\right|_{\mathcal{S}_1}\right)\epsilon_2 +
 \abs{f_2(\tilde{\Phi}_1\left.(A)\right|_{\mathcal{S}_1})-f_2(\Phi_1(A))}\\
 \times\Big(k_1(A)\epsilon_1+\abs{f_1(A)}\norm{A}\Big)
  + \abs{f_2(\Phi_1(A))}k_1(A)\epsilon_1.
\end{multlined}
\end{equation}

Since $f_2$ is Lipschitz on any compact subset,
\begin{equation}
\begin{multlined}
\norm{\tilde{\Phi}_2\left.\left(\tilde{\Phi}_1\left.(A)\right|_{\mathcal{S}_1}\right)\right|_{\mathcal{S}_2} - \Phi_2\circ \Phi_1(A)}
\\
 \leq k_2\left(\tilde{\Phi}_1\left.(A)\right|_{\mathcal{S}_1}\right)\epsilon_2 +  L_2 k_1(A)\epsilon_1\Big(k_1(A)\epsilon_1 +\abs{f_1(A)}\norm{A}\Big) \\
 + \abs{f_2(\Phi_1(A))}k_1(A)\epsilon_1,
\end{multlined}
\end{equation}
and all terms on the right hand size are of order $\epsilon_1$ or $\epsilon_2$.
One choice of $\epsilon$ and $k(A)$ is then,
\begin{align}
&\epsilon  = \epsilon_1+\epsilon_2\\
&k(A) = \begin{multlined}k_2\left(\tilde{\Phi}_1\left.(A)\right|_{\mathcal{S}_1}\right) + L_2k_1(A)^2\epsilon_1 \\
+ L_2 \abs{f_1(A)}\norm{A}k_1(A)+\abs{f_2(\Phi_1(A))}k_1(A).\end{multlined}
\end{align}

The scaling of $\eta$ is simplified by the definition of the subspace $\mathcal{S}$, since $\Phi_1$/$\Phi_2$ are $\eta_1/\eta_2$ close to local dualities $\Phi_1'/\Phi_2'$.
Therefore by \cref{lm cubitt2} and triangle inequality we have $\norm{V-V'}\leq \eta_1 + \eta_2$.
\end{proof}

\subsection{Physical properties}\label{sect Physical properties}

This section walks through how the parameters in the definition of approximate and approximately-local duality translates to different physical properties.

\subsubsection{Measurement outcomes}\label{sect Measurement outcomes}
\enlargethispage{2\baselineskip}
\cref{Measurement duality axioms} includes a spectrum preserving statement motivated by considering that dual measurement outcomes should be related.
This included a scaling factor relating the spectra which is associated with a possible unit rescaling.
Now considering approximate duality maps, the rescaled eigenvalues of corresponding observables are approximately equal with a controlled error.

\begin{restatable}[Approximate eigenvalues]{prop}{propapproximateeigenvalues}
\label{prop approximate eigenvalues}
Let the Hermitian operator $A$ act on $\left(\mathbb{C}^d \right)^{\otimes n}$ and $\tilde{\Phi}$ be a $(\mathcal{S},\epsilon,\eta)$- approximately local duality map.
Let $\lambda_i(A)$, $\lambda_i(\tilde{\Phi}(A)|_{\mathcal{S}})$ be the i'th smallest eigenvalues of $A$ and $\tilde{\Phi}(A)|_{\mathcal{S}}$ respectively.
Then for all $1\leq i \leq d^n$ and all $j$ such that $(i-1)(p+q)+1\leq j\leq i(p+q)$,
\begin{equation}
\abs{\lambda_j(\tilde{\Phi}(A)|_{\mathcal{S}}) - f(A)\lambda_i(A) } \leq k(A)\epsilon.
\end{equation}
Where the integers $p$, $q$ and $f(\cdot)$ is the function appearing the corresponding exact duality map.
\end{restatable}
\begin{proof}
Let $\Phi$ be the exact duality map which is $\epsilon$-close to the restricted $\tilde{\Phi}$ as in \cref{eqn closeness cond} and $\eta$-close to the local duality.
For any $i$, $j$ satisfying the above inequalities, $\lambda_j\left( \Phi(A)\right) = f(A)\lambda_i(A)$ from axiom (iii) of
\cref{Measurement duality axioms} of exact dualities.
Combining \cref{eqn closeness cond} with Weyl's inequality ($\abs{\lambda_j(A)-\lambda_j\left(B \right)}\leq \norm{A- B}$) gives,
\begin{align}
\abs{\lambda_j(\tilde{\Phi}(A)|_{\mathcal{S}})-f(A)\lambda_i(A)} & = \abs{\lambda_j(\tilde{\Phi}(A)|_{\mathcal{S}})-\lambda_j\left(\Phi(A) \right)}\\
& \leq \norm{\tilde{\Phi}(A)|_{\mathcal{S}} - \Phi(A)}\\
& \leq k(A)\epsilon.
\end{align}
\end{proof}

\subsubsection{Thermal properties}\label{sect Thermal properties}

Similarly \cref{Thermal duality axioms} includes a partition-function-like statement motivated by requiring dual thermal properties.
Approximate duality mappings preserve partition functions of a given Hamiltonian up to a controllable error, when the restricted subspace is taken to be the low-energy subspace of the Hamiltonian in question.

\begin{restatable}[Approximate partition functions]{prop}{propapproximatepartitionfunctions}
\label{prop approximate partition functions}
Let the Hamiltonian $H$ act on $\left(\mathbb{C}^d \right)^{\otimes n}$ and $\tilde{\Phi}$ be the $(\mathcal{S},\epsilon, \eta)$-duality map into $(\mathbb{C}^{d'} )^{\otimes m}$, where $\mathcal{S}$ is the low energy subspace of $H$ with energy less than $\Delta$.
The relative error in the dual partition functions is given by,
\begin{equation}
\frac{\abs{\mathcal{Z}_{\tilde{\Phi}(H)}(\beta) - (p+q)\mathcal{Z}_{H}(f(H)\beta)}}{(p+q)\mathcal{Z}_{H}(f(H)\beta)} \leq \frac{(d')^m e^{-\beta \Delta}}{(p+q)d^n e^{-\beta f(H)\norm{H}}} + \left(e^{\beta k(H)\epsilon}-1 \right),
\end{equation}
where the integers $p$, $q$ and $f(\cdot)$ is the function in the corresponding exact duality map.
\end{restatable}

\begin{proof}
By axiom (iii) \cref{Thermal duality axioms} of an exact duality \linebreak $(p+q)\trace \left[e^{-\beta f(H)H} \right] = \trace \left[ e^{-\beta \Phi(H)}\right] $.
Therefore,
\begin{align}
\frac{\abs{\mathcal{Z}_{\tilde{\Phi}(H)}(\beta) - (p+q)\mathcal{Z}_{H}(f(H)\beta)}}{(p+q)\mathcal{Z}_{H}(f(H)\beta)} & = \frac{\abs{\trace\left[e^{-\beta \tilde{\Phi}(H)} \right]- (p+q)\trace \left[e^{-\beta f(H)H} \right]}}{(p+q)\trace \left[e^{-\beta f(H)H} \right]}\\
& = \frac{\abs{\trace\left[e^{-\beta \tilde{\Phi}(H)} \right]- \trace \left[e^{-\beta \Phi(H)} \right]}}{(p+q)\trace \left[e^{-\beta f(H)H} \right]}\\
& \begin{multlined}
\leq \frac{ \abs{\trace\left[e^{-\beta \tilde{\Phi}(H)} \right]- \trace \left[ e^{-\beta  \left.\tilde{\Phi}(H)\right|_{\mathcal{S}}  }\right]}}{(p+q)\trace \left[e^{-\beta f(H)H} \right]}\\
+ \frac{ \abs{\trace \left[ e^{-\beta  \left.\tilde{\Phi}(H)\right|_{\mathcal{S}}  }\right] - \trace\left[e^{-\beta \Phi(H)} \right]}}{\trace \left[e^{-\beta \Phi(H)} \right]} .
\end{multlined}
\end{align}
Bounding the numerator and denominator of the first term:
\begin{align}
 \abs{\trace\left[e^{-\beta \tilde{\Phi}(H)} \right]- \trace \left[ e^{-\beta  \left.\tilde{\Phi}(H)\right|_{\mathcal{S}}  }\right]} &\leq (d')^m e^{-\beta \Delta}.\label{eqn num bound}\\
\trace \left[e^{-\beta f(H)H} \right]&\geq d^n e^{-\beta  f(H)\norm{H}}. \label{eqn denom bound}
\end{align}
The second term is bounded by considering eigenvalues.
Let $\lambda_l$ be the $l$'th eigenvalue of $\left.H'\right|_{\mathcal{S}} $ in non-decreasing order.
Then by the argument in \cref{prop approximate eigenvalues} the $l$'th eigenvalue of $\Phi(H)$ (in the same order) is given by $\lambda_l + k(H)\epsilon_l$ where $\abs{\epsilon_l}\leq \epsilon$ for all $l$.
Hence,
\begin{align}
 \abs{\trace \left[ e^{-\beta  \left.\tilde{\Phi}(H)\right|_{\mathcal{S}}  }\right] - \trace\left[e^{-\beta \Phi(H)} \right]} &\leq \sum_l \abs{e^{-\beta \lambda_l} - e^{-\beta(\lambda_l + k(H)\epsilon_l)}}\\
& =\sum_l e^{\beta(\lambda_l+k(H)\epsilon_l)} \abs{e^{\beta k(H)\epsilon_l}-1}\\
& \leq (e^{\beta k(H)\epsilon}-1)\trace \left[e^{-\beta \Phi(H)} \right].
\end{align}
Combining the above with \cref{eqn num bound} and \cref{eqn denom bound} gives the result.
\end{proof}

\subsubsection{Time dynamics}\label{sect Time dynamics}

\cref{defn state map} demanded consistent time dynamics for exact duality mappings as a constraint to specify the form of the corresponding state map.
As expected when considering approximate duality maps this statement is relaxed, such that time dynamics of the two systems is close up to an error that increases with time.

\begin{restatable}[Approximate time dynamics]{prop}{propapproximatetimedynamics}
\label{prop approximate time dynamics}
Let $\tilde{\Phi}$ be a $(\mathcal{S},\epsilon,\eta)$-approximately local duality map with corresponding exact duality $\Phi(\cdot) = f(\cdot)\mathcal{E}(\cdot)$.
Given a Hamiltonian $H$ such that $\mathcal{S}$ is the low energy subspace with eigenvalues $< \Delta$.
Then for any density matrix $\rho$ in the encoded subspace, such that $\Phi(\mathbb{I})\rho = \rho$, the time dynamics of the approximate duality mapping is close to that of the exact mapping:
\begin{equation}
\norm{e^{-i\tilde{\Phi}(H)t}\rho e^{i\tilde{\Phi}(H)t} - e^{-i\Phi(H)t}\rho e^{i\Phi(H)t}}_1 \leq 2\epsilon k(H) t + \eta.
\end{equation}
\end{restatable}

This follows from an identical argument as Proposition 29 from \cite{Cubitt2019}, applying instead $\norm{\left.\tilde{\Phi}(H)\right|_{\mathcal{S}} - \Phi(H)}\leq k(H)\epsilon$ at the final step.

\section*{Acknowledgements}
TSC~is supported by the Royal Society.
HA~is supported by EPSRC DTP Grant Reference: EP/N509577/1 and EP/T517793/1.
We thank Nikolas Breuckmann for interesting discussions that led to the initial idea for this work and QIP 2023 reviewers for helpful comments and suggestions.
This work has been supported in part by the EPSRC Prosperity Partnership in Quantum Software for Simulation and Modelling (grant EP/S005021/1), and by the UK Hub in Quantum Computing and Simulation, part of the UK National Quantum Technologies Programme with funding from UKRI EPSRC (grant EP/T001062/1).

\begin{appendices}
\crefalias{section}{appendix}

\section{Proof of \cref{lm entropy of mixtures of mixed}}\label{appen proof 1}

\lmentropyofmixturesofmixed*

\begin{proof}
Write each mixed state as a sum of pure states:
\begin{equation}
\rho_x = \sum_{j=1}^m \lambda_j^{(x)} \ket{\phi_j^{(x)}}\bra{\phi_j^{(x)}},
\end{equation}
where $\{ \ket{\phi_j^{(x)}}\}_{j=1}^m$ form an orthogonal basis for a given $x$, but in general $\braket{\phi_i^{(x)}|\phi_j^{(y)}}\neq 0 $ for $x\neq y$.
The full density operator with these expansions reads:
\begin{equation}
\rho_\mathcal{A} = \sum_{x=1}^k \sum_{j=1}^m p_x \lambda_j^{(x)} \ket{\phi_j^{(x)}}\bra{\phi_{j}^{(x)}}.
\end{equation}
Introduce a Hilbert space, $\mathcal{H}_\mathcal{R}$, with $\dim(\mathcal{H}_\mathcal{R})= mk$ and an orthonormal basis labeled by $\ket{xj}_\mathcal{R}$.
Consider a purification of $\rho_\mathcal{A}$,
\begin{equation}
\ket{\mathcal{AR}} = \sum_{x,j}\sqrt{p_x \lambda_j^{(x)}}\ket{\phi_j^{(x)}}_\mathcal{A} \otimes \ket{xj}_\mathcal{R},
\end{equation}
where
\begin{align}
\rho_\mathcal{A} &= \trace_\mathcal{R} \left[\ket{\mathcal{AR}}\bra{\mathcal{AR}} \right] = \sum_{x,j} p_x \lambda_j^{(x)} \ket{\phi_j^{(x)}} \bra{\phi_j^{(x)}},
\end{align}
and
\begin{align}
\rho_\mathcal{R} &= \trace_\mathcal{A} \left[\ket{\mathcal{AR}}\bra{\mathcal{AR}} \right]\\
& = \sum_{x,j,x',j'} \sqrt{p_x \lambda_j^{(x)}} \sqrt{p_{x'}\lambda_{j'}^{(x')}} \braket{\phi_j^{(x)}|\phi_{j'}^{(x')}}_\mathcal{A} \ket{xj}\bra{x'j'} \label{eqn rho R}.
\end{align}
Also define
\begin{equation}
\rho_\mathcal{R'} := \sum_{x,j} p_x \lambda_j^{(x)} \ket{xj}\bra{xj}.\label{eqn rho R'}
\end{equation}
The relative entropy between the two reservoir states is given by
\begin{align}
S(\rho_\mathcal{R}||\rho_\mathcal{R'}) &:= \trace \rho_\mathcal{R} \log \rho_\mathcal{R} - \trace \rho_\mathcal{R}\log \rho_\mathcal{R'}\\
& = - S(\rho_\mathcal{R}) - \trace \rho_\mathcal{R}\log \rho_\mathcal{R'}\\
& = - S(\rho_\mathcal{A}) - \trace \rho_\mathcal{R}\log \rho_\mathcal{R'}.
\end{align}
Where the last line uses $S(\rho_\mathcal{R}) = S(\rho_\mathcal{A})$.
Since $\ket{xj}$ forms an orthogonal basis $\log \rho_\mathcal{R'} = \sum_{x,j}\log \left(p_x \lambda_j^{(x)} \right) \ket{xj}\bra{xj}$.
Further algebraic manipulation of the last term results in,
\begin{align}
\trace \left[ \rho_\mathcal{R} \log \rho_\mathcal{R'}\right] &= \sum_{j,x} \log \left(p_x \lambda_j^{(x)} \right) \trace \left(\rho_\mathcal{R}\ket{xj}\bra{xj} \right)\\
& = \sum_{x,j} \log \left(p_x \lambda_j^{(x)} \right) \braket{xj| \rho_\mathcal{R}|xj}\\
& = \sum_{x,j} p_x \lambda_j^{(x)} \log \left(p_x \lambda_j^{(x)} \right)\\
& = \sum_{x,j} p_x \lambda_j^{(x)}\log p_x + \sum_{x,j} p_x \lambda_j^{(x)}\log(\lambda_j^{(x)})\\
& = \sum_x p_x \log p_x + \sum_x p_x \sum_j \lambda_j^{(x)} \log \lambda_j^{(x)}\\
& = \sum_x p_x \log p_x - \sum_x p_x S(\rho_x).
\end{align}
We arrive at an expression for the entropy of our mixture of mixed states,
\begin{equation}
S(\rho_\mathcal{A}) = \sum_x p_x S(\rho_x) - \sum_x p_x \log p_x - S(\rho_\mathcal{R}||\rho_\mathcal{R'}).
\end{equation}
Since the relative entropy $S(\rho_\mathcal{R}||\rho_\mathcal{R'})=0$ if and only if $\rho_\mathcal{R} = \rho_\mathcal{R'}$, the expressions for $\rho_\mathcal{R}$ and $\rho_\mathcal{R'}$ in \cref{eqn rho R}, \cref{eqn rho R'} respectively, imply that the two density matrices are equal if and only if the corresponding vectors $\ket{\phi_j^{(x)}}$ form an orthogonal set (given $j,x$ such that $\lambda_j^{(x)}\neq 0$).
This is equivalent to stating that the mixed states $\rho_x$ must have orthogonal support.
\end{proof}

\section{Proof of \cref{lm orthogonal density matrices}}\label{appen proof 2}

\lmorthogonaldensitymatrices*

\begin{proof}
Any state in $\mathcal{S}(\mathcal{H})$ can be written as a linear combination of the set of pure states.
The map $\phi$ obeys entropy relation~(a) so,
\begin{equation}
S\left(\phi(\sum_{i=1}^d \lambda_i \sigma_i) \right)  = S\left(\sum_{i=1}^d \lambda_i \sigma_i \right) + \log \alpha.
\end{equation}
Since $\{ \sigma_i\}_{i=1}^d$ have orthogonal support, \cref{lm entropy of mixtures of mixed} can be applied to the first term:
\begin{equation}
S\left(\phi(\sum_{i=1}^d \lambda_i \sigma_i) \right) = \sum_{i=1}^d \lambda_i S(\sigma_i) - \sum_{i=1}^d \lambda_i \log \lambda_i + \log \alpha.
\end{equation}
Reusing the entropy preserving property of $\phi$, this time with a sum over pure states with $S(\sigma_i)=0$, $S(\phi(\sigma_i))=\log \alpha$ for all $i$.
Since $\sum_{i=1}^d \lambda_i = 1$, $\sum_{i=1}^d \lambda_i S(\phi(\sigma_i))=\log \alpha$, thus
\begin{equation}
S\left(\phi(\sum_{i=1}^d \lambda_i \sigma_i) \right)= - \sum_{i=1}^d \lambda_i \log \lambda_i + \sum_{i=1}^d \lambda_i S(\phi(\sigma_i)).
\end{equation}
Since there is an equality, the only if direction of \cref{lm entropy of mixtures of mixed} implies that $\{ \phi(\sigma_i)\}$ must have orthogonal support.
\end{proof}

\section{Vanishing off-diagonal matrix elements}
\label{Vanishing off-diagonal matrix elements}

This appendix demonstrates that if a Hermitian projector, $\Pi = \left(
\begin{array}{c|c}
A & B \\
\hline
C & D
\end{array}
\right) $, has $D=\mathbf{0}$ then necessarily $B=C=\mathbf{0}$.

To show this, two properties of $\Pi$ are useful,
\begin{enumerate}
\item Projectors are idempotent, $\Pi^2 = \Pi$:
\begin{equation}
\left(
\begin{array}{cc}
A & B \\
C & 0
\end{array}
\right) \left(
\begin{array}{cc}
A & B \\
C & 0
\end{array}
\right)  =
\left(
\begin{array}{cc}
A^2 + BC & AB \\
CA & CB
\end{array}
\right) = \left(
\begin{array}{cc}
A & B \\
C & 0
\end{array}
\right),
\end{equation}
therefore $CB=\mathbf{0}$.
\item Hermitian operators are self-adjoint, $\Pi^\dagger = \Pi$:
\begin{equation}
\left(
\begin{array}{cc}
A^\dagger & C^\dagger \\
B^\dagger & 0
\end{array}
\right) = \left(
\begin{array}{cc}
A & B \\
C & 0
\end{array}
\right),
\end{equation}
therefore $B^\dagger = C$.
\end{enumerate}
Putting these together gives $BB^\dagger=\mathbf{0}$. The eigenvalues and eigenvectors of adjoint matrices are related by $Bv_i=\lambda_i v_i$ and $B^\dagger v_i = \overline{\lambda_i}v_i$. Therefore since $|\lambda_i|=0$ for all $i$ in the diagonal basis, $B=C=\mathbf{0}$.

\section{Matching up of spectra}\label{appen matching}

The proof of \cref{lm Constrained scale functions} claims that the trivial permutation $\sigma(k)=k$ is the only allowed case for
\begin{equation}
\frac{G(A)}{f(A)}  h(\mu_{\sigma(j)}) g(c_{\sigma(j)}P_{\sigma(j)}) f(c_{\sigma(j)}P_{\sigma(j)})c_{\sigma(j)} = \mu_jc_j .
\end{equation}
Note that there are $\dim(A)$ such equations corresponding to selecting for the different spectral projectors.

One could conceive that the permutation depends on the operator as well as the map, so that for a different operator $B$,
\begin{equation}
\frac{G(B)}{f(B)} h\left(\mu^{(B)}_{\nu(j)}\right) g \left( c^{(B)}_{\nu(j)}P^{(B)}_{\nu(j)}\right) f\left(c^{(B)}_{\nu(j)}P^{(B)}_{\nu(j)}\right)c^{(B)}_{\nu(j)} = \mu^{(B)}_j c^{(B)}_j,
\end{equation}
where $\sigma(j)\neq\nu(j)$.
However, the spectral projectors are analytic functions of the matrix~\cite{Kato} (assuming for now that the matrices are non-degenerate).
Any two matrices $A$ and $B$ can be connected by a smooth path, ruling out a change in the permutation which would require a discontinuous jump.
Therefore the permutation must be consistent for all inputs to the map.
The exception to this case is where there are degeneracies, but as we are only interested in equating eigenvalues, a permutation of degenerate eigenvalues within the degenerate spectral projectors has no affect.

To justify that the permutation must be trivial everywhere we show by contradiction that it must be trivial for any $A$.
Consider the following analytic change to the operator $A\rightarrow A'$: one eigenvalue is changed $\mu_k'c_k'= \mu_kc_k+\delta$ ($\delta>0$), whilst all other eigenvalues $\mu_{j\neq k}'c_{j\neq k}'=\mu_{j\neq k}c_{j\neq k}$, and all spectral projectors $P'_{j} = P_j$ are held unchanged.
For this new operator, a similar set of equations hold with the same permutation:
\begin{equation}
\frac{G(A')}{f(A')} h(\mu'_{\sigma(j)}) g( c'_{\sigma(j)}P_{\sigma(j)})f(c'_{\sigma(j)}P_{\sigma(j)})c'_{\sigma(j)} = \mu'_j c'_j.
\end{equation}
For all $j\neq1$ such that $\sigma(j)\neq 1$
\begin{equation}
\frac{G(A)}{f(A)}\frac{f(A')}{G(A')} = \frac{h(\mu_{\sigma(j)})g( c_{\sigma(j)}P_{\sigma(j)})f(c_{\sigma(j)}P_{\sigma(j)})c_{\sigma(j)}}{h(\mu_{\sigma(j)}) g( c_{\sigma(j)}P_{\sigma(j)})f(c_{\sigma(j)}P_{\sigma(j)})c_{\sigma(j)}} =1.
\end{equation}

Assume for contradiction that $\sigma(k)= j \neq k$, so that there is some non-trivial permutation and dual spectral projectors are not paired with eigenvalues related by $f(A)$.
Then
\begin{align}
\frac{G(A)}{f(A)}h(\mu_j) g(c_jP_j)f(c_jP_j)c_j & = \mu_kc_k\\
\frac{G(A')}{f(A')}h(\mu_j) g(c_jP_j)f(c_jP_j)c_j & = \mu_kc_k + \delta.
\end{align}
Therefore, $\frac{G(A')}{f(A')}\frac{f(A)}{G(A)}\mu_kc_k = \mu_kc_k + \delta$ and $\delta =0$, contradicting $\delta>0$ and the trivial permutation is the only allowed case.

\section{Kramers-Wannier duality}\label{appen K-W}

The Kramer-Wannier duality links two 2d Ising models, one at low temperature (strong interaction strength) with another at high temperature (weak interaction strength).
The duality is identified by computing the partition function of both systems in their respective limits.
This appendix outlines the Kramer-Wannier duality and how it arises, based on \cite{Baxter1989,MITStat}.
Final we show explicitly how it lies outside the original simulation framework of \cite{Cubitt2019} and how can be placed in our more general framework of duality.

\subsection{Low temperature expansion}
The Ising model on an $N$ site lattice is governed by the Hamiltonian $ H = -J \sum_{\langle i, j \rangle}\sigma_i \sigma_j$.
Consider the isotropic case where the interaction strength $J$ is the same across both horizontal and vertical directions, $K:=\beta J$.
If $K>0$ the model is ferromagnetic and the ground state will have all spins aligned.
In the low temperature regime the system is dominated by its ground state.
The expansion for the partition function at low temperature is given by the ground state configuration plus low energy fluctuations -- i.e. 1, 2, 3, ... spins aligned anti-parallel.
The additional energy cost of one flipped spin in a 2d lattice is $4\times 2K$ and that of two flipped spins in a block is $6 \times 2K$.
Counting the degeneracies of these states the partition function is given by,
\begin{equation}
Z \approx 2e^{\text{no. bonds total} \times K} \left[1 + Ne^{-4\times 2K} + 2N e^{-6 \times 2K} + ... \right].
\end{equation}
The energy cost comes from the domain wall boundary between the regions of anti-parallel spins.
In the thermodynamic limit $N\rightarrow \infty$ the multiplicities become insignificant and the partition function can be written as
\begin{equation}
Z \approx 2e^{\text{no. bonds total} \times K} \sum_{\text{islands of -ve spins}} e^{-2K \times \text{perimeter of island}}.
\end{equation}
The terms in this summation can be represented graphically by creating islands of increasingly large regions of anti-aligned spin.

\subsection{High temperature expansion}
The high temperature expansion starts instead with independent spins and the partition function is expanded in powers of $\beta$.
A convenient simplification is to expand in powers of $\tanh K$ instead.
This is equivalent to doing a high temperature expansion since $\tanh K$ is less than 1 (except when $\beta \rightarrow \infty$) so in the high temperature region powers of $\tanh K$ are increasingly small.
Since $(\sigma_i\sigma_j)^2=1$ the bond $\langle i,j\rangle$ Boltzmann factor can be rewritten as
\begin{align}
e^{K\sigma_i\sigma_j} & = \frac{e^{K}+e^{-K}}{2} + \frac{(e^{K}-e^{-K})}{2}\sigma_i \sigma_j\\
& = \cosh K (1+\tanh K \sigma_i \sigma_j).
\end{align}
Applying this transformation to the partition function gives
\begin{align}
Z&= \sum_{\{\sigma_i \}} e^{K\sum_{\langle i, j \rangle}\sigma_i \sigma_j}\\
& = (\cosh K)^\text{no. bonds} \sum_{\{ \sigma_i\}} \prod_{\langle i, j \rangle} (1+\tanh K \sigma_i \sigma_j).
\end{align}
The product $\prod_{\langle i,j \rangle} (1+\tanh K \sigma_i \sigma_j) = (1+\tanh K \sigma_a \sigma_b)(1+\tanh K \sigma_a \sigma_c)(1+\tanh K \sigma_a \sigma_d)...$ generates $2^{N_b}$ terms where $N_b$ is the number of bonds in the lattice.
Each term in the sum can again be represented as a graph: for all edges in the lattice draw a line on the edge $(i,j)$ if there is a factor of $\tanh K \sigma_i \sigma_j$ (an occupied bond) and draw no line if the term in the expansion is $1$ (an unoccupied bond).
Each term in the expansion of the product is of the form
\begin{equation}
(\tanh K )^\text{no. occupied bonds}\sigma_1^{p_1} \sigma_2^{p_2} \sigma_3^{p_3}... \sigma_N^{p_N}.
\end{equation}
Now we perform the sum over each spin being $\pm 1$.
Summing over $\sigma_i$ gives a factor of $2$ if $p_i$ even and $0$ if $p_i$ odd.
Therefore only graphs where every site has an even number of occupied legs is non-vanishing.
These graphs from closed paths on the lattice.
The high temperature series expansion is then given by
\begin{equation}
Z = 2^N \times (\cosh K)^\text{no. bonds total} \sum_\text{closed graphs} (\tanh K)^\text{no. occupied bonds in the graph}.
\end{equation}

\subsection{Free energy duality}

Taking the thermodynamic limit so $N\rightarrow \infty$ the total number of bonds in the lattice becomes $\approx 2N$.
Make the couplings strengths for the models in the different temperature regimes distinct: $K$ in the high temperature expansion, $\tilde{K}$ in the low temperature expansion.
The duality is identified by comparing the low and high temperature series expansions:
\begin{align}
\text{Low temp:} \quad &Z(\tilde{K}) = 2e^{2N\tilde{K}}\sum_\text{islands of -ve spin} e^{-2\tilde{K} \times \text{perimeter of islands}}\\
\text{High temp:} \quad & Z(K) = 2^N (\cosh K)^{2N} \sum_{\text{closed graphs}} \tanh K^{\text{length of graph}}.
\end{align}

There is a correspondence between the two sums since islands of sites can be considered as closed graphs and vice versa.
They differ only at the boundaries, but in the thermodynamic limit this difference becomes negligable.
Defining the function,
\begin{equation}
g(x) := \lim_{N\rightarrow \infty} \ln \sum_\text{closed graphs} x^\text{no. lines in the graph},
\end{equation}
the arguments of $g$ in each of the above sum are related by the duality condition
\begin{equation} \label{appen duality condition}
e^{-2\tilde{K}} \leftrightarrow \tanh K \qquad \tilde{K} = D(K) = -\frac{1}{2} \ln \tanh K.
\end{equation}

With the above function $g$ the free energies per particle can be written as:
\begin{align}
\text{Low temp:} \quad &-\beta f_H = \frac{\ln Z(K)}{N} = 2K + \ln \left\{e^{-4\times 2K} + 2e^{-4\times 6K} - \frac{5}{2} e^{-8\times 2K}+...\right\}\\
\text{High temp:} \quad & -\tilde{\beta} f_{\tilde{H}} = \frac{\ln Z(\tilde{K})}{N} = \ln 2 + 2 \ln \cosh K + \ln \left\{(\tanh K)^4+ ...\right\}.
\end{align}
The duality condition, \cref{appen duality condition}, then relates the two free energies by:
\begin{equation}
\tilde{\beta} f_{\tilde{H}} = \beta f_H + 2 \beta J - \ln \left[2 \cosh^2\left(\tilde{\beta}\tilde{J}\right) \right].
\end{equation}
Some algebra manipulates the free energy relation into a simpler form:
\begin{align}
2\beta J - \ln \left[2 \cosh^2\left(\tilde{\beta}\tilde{J}\right) \right] & = \ln \left[\frac{e^{2\beta J}}{2\cosh^2\left(\ln \tanh^{-1/2}J\beta \right) } \right]\\
&= \ln \left[\frac{e^{2\beta J}}{2\times (1/4) \times \left(\tanh^{-1/2}J\beta + \tanh^{1/2}J\beta \right)^2 } \right]\\
&= \ln \left[\frac{e^{2\beta J}}{2\times (1/4) \times \left(\tanh^{-1} J\beta + \tanh J\beta + 2\right) } \right]\\
&= \ln \left[\frac{e^{2\beta J}}{(1/2) \times \left( \frac{\cosh J\beta}{\sinh J\beta} + \frac{\sinh J\beta}{\cosh J\beta} + 2\right)  } \right]\\
&= \ln \left[\frac{e^{2\beta J}}{ \left(\frac{\cosh^2 J\beta + \sinh^2 J\beta + 2sinh J\beta \cosh J \beta }{2 \sinh J\beta \cosh J \beta}\right)  } \right]\\
&= \ln \left[\frac{e^{2\beta J}}{\left(\frac{\cosh 2J\beta + \sinh 2J\beta }{\sinh 2J\beta}\right)  } \right]\\
&= \ln \left[\frac{e^{2\beta J}}{\left(\frac{e^{2J\beta} }{\sinh 2J\beta}\right)  } \right]\\
&= \ln \left[\sinh 2J\beta\right]
\end{align}

From the free energy the partition functions can be related using $-\beta f = \ln Z$:
\begin{align}
Z_{\tilde{H}}(\tilde{\beta}) &= \exp \left[-N\tilde{\beta}f_{\tilde{H}} \right]\\
& = \exp \left[-N\beta f_H - N \ln \sinh (2\beta J) \right]\\
& = e^{\ln\sinh(2\beta J)^{-N}}Z_H(\beta)\\
& = \left[ \sinh (2\beta J)\right]^{-N} Z_H(\beta).
\end{align}

\subsection{In this framework}

We consider the map on operators that relates the two dual Hamiltonians.
The encoding part of the duality is simple as the Hilbert space is the same size so there are no copies ($p=1$) and the form of the operators is the same so the unitary is simply the identity.
The more interesting part of the duality appears in the scale factor, which should be a function of the initial Hamiltonian only.
In the partition function and time evolution operator, temperature and the Hamiltonian always appear as a product ($\beta H$).
Since in the Ising model the Hamiltonian is proportional to the coupling constant $J$, for Ising type Hamiltonians there is a trivial duality condition $J\beta = J'\beta'$.
We therefore have additional freedom in how we chose to construct the set of maps that correspond to different physical scenarios if one were to engineer this duality.
The first choice is consistent with how the duality framework in this paper has been set out, however the different approaches are mathematically equivalent.

In the first instance we will view Kramer-Wannier through the lens of a strong-weak duality: equating the temperatures of the dual systems $\beta = \tilde{\beta}$.
A strongly interacting Ising model with interaction strength $J$ is dual to a weakly interacting Ising model with interaction strength $\tilde{J}\neq J$ at the same temperature.
This leads to a non-trivial scaling function for the map on operators that depends both on the operator and the temperature of the system:
\begin{align}\label{K-W Hamiltonian map}
\Phi_H(H) & = -\frac{1}{2J\beta}\ln\tanh(J\beta) \times \mathbb{I}(H)\\
& = f(H, \beta) \times \mathcal{E}(H),
\end{align}
where $\mathcal{E}(H)=\mathbb{I}(H)$ is an encoding satisfying the axioms 1-3 from \cref{thm Encodings} and $f(H, \beta) = -\frac{1}{2J\beta}\ln\tanh(J\beta)$.
The coupling strength, $J$, can be written as a function of the Hamiltonian norm and $n$ the number of lattice sites: $J= \frac{\norm{H}}{2n(n-1)}$.

Another approach could be to fix the interaction strength $J = \tilde{J}$ and consider a high-low temperature duality where physical properties of the two dual systems are evaluated at different temperatures $\beta \neq \tilde{\beta}$.
This is does not allow the duality to be manipulated into our framework since we do not allow a temperature map.
Here the map on operators is independent of temperature with a trivial scaling function, $f(H)=1$:
\begin{equation}
\Phi(H) = \mathbb{I}(H).
\end{equation}
This viewpoint introduces the necessity of a temperature map, $\Phi_\beta$, should map positive reals to positive reals and be compatible with the Hamiltonian map such that the duality condition is satisfied.
We will allow the temperature map to additionally depend on Hamiltonian parameters so there is a consistent set of maps for one system.
In order to satisfy \cref{condition eqn}, the temperature map is
\begin{equation}\label{K-W temperature map}
\Phi_\beta(\beta)=-\frac{1}{2J}\ln \tanh(J\beta).
\end{equation}

This is perhaps the more immediate viewpoint from the Kramer-Wannier literature but both approaches are mathematically equivalent.
In fact using any interpolation of these two cases is also valid.
We could consider two Ising models with different interaction strengths at different temperatures, as long as the product obeys \cref{condition eqn}.
Furthermore neither case fit into the original simulation framework in \cite{Cubitt2019} for differing reasons.
The first has a non-trivial scaling function so that the spectra of two dual operators is not equal.
The second has a non-trivial temperature map so that the systems are only dual if considered at the appropriate temperature.

We can complete the description by providing a compatible map on states.
Again we have choices.
We could require the Born rule with respect to energy measurements should be preserved, or we could alternatively demand that thermal states map to thermal states.
Starting with energy measurement outcomes, the expected behaviour
\begin{equation}\label{eqn born cond}
\trace\left[H \rho \right] = \frac{1}{f(H)} \trace \left[\Phi_H(H)\Phi_\text{state}(\rho) \right]
\end{equation}
is achieved by a trivial mapping on states $\Phi_\text{states}(\rho)=\rho$.
If instead we propose preserving Gibbs states,
\begin{equation}
\Phi_\text{state}\left(\frac{e^{-\beta H}}{Z_H(\beta)}\right) = \frac{e^{-\Phi_\beta(\beta)\Phi_H(H)} }{Z_{\Phi_H(H)}(\Phi_\beta(\beta))},
\end{equation}
then another choice for a map on states is $\Phi_\text{state}(\rho) = \frac{\epsilon(\rho)}{\trace(\epsilon(\rho))}$ with $\epsilon(\rho) = \rho^{\frac{1}{2J\beta}\ln \tanh(J\beta)}$.
These state mappings will preserve measurement outcomes and thermal states respectively paired with either the strong-weak or high-low formulations described earlier.

\end{appendices}

\printbibliography

\end{document}